\documentclass[12pt,onecolumn,draftcls,letterpaper]{IEEEtran}
  \usepackage{amsmath,amssymb,latexsym,epsfig,graphics,color}

\newtheorem{lemma}{Lemma}

\newcommand\zerob{\boldsymbol{0}}
\newcommand\bb{\boldsymbol{b}}
\newcommand\cb{\boldsymbol{c}}
\newcommand\hb{\boldsymbol{h}}
\newcommand\mb{\boldsymbol{m}}
\newcommand\nb{\boldsymbol{n}}
\newcommand\qb{\boldsymbol{q}}
\newcommand\yb{\boldsymbol{y}}

\newcommand\Ab{\mathbb{A}}
\newcommand\Bb{\mathbb{B}}
\newcommand\Db{\mathbb{D}}
\newcommand\Eb{\mathbb{E}}
\newcommand\Gb{\mathbb{G}}
\newcommand\Hb{\mathbb{H}}
\newcommand\Ib{\mathbb{I}}
\newcommand\Pb{\mathbb{P}}
\newcommand\Wb{\mathbb{W}}

\newcommand\ac{{\mathcal A}}
\newcommand\cc{{\mathcal C}}

\newcommand\vectorize{\mathrm{vec}}
\newcommand\tr{\mathrm{tr}}
\newcommand{\dis}{\displaystyle}
\newcommand{\oneb}{{\bf 1}}

\newcounter{step}

\begin{document}
\title{Maximum-Likelihood Priority-First Search Decodable Codes for Combined Channel
       Estimation and Error Protection}
\author{Chia-Lung Wu,~\IEEEmembership{Student Member,~IEEE,}
        Po-Ning~Chen,~\IEEEmembership{Senior Member,~IEEE,}
        Yunghsiang~S.~Han,~\IEEEmembership{Member,~IEEE,}
        and
        Ming-Hsin Kuo
\thanks{Submission to IEEE Transactions on Information Theory.  {\bf Prof.~P.-N.~Chen will be the corresponding author for this submission}.
This work was supported by the NSC of Taiwan, ROC, under grants NSC 95-2221-E-009-054-MY3.}
\thanks{C.-L.~Wu is currently studying toward his Ph.D.~degree under the joint advisory of Profs.~P.-N.~Chen
        and Y.~S.~Han (E-mail: clwu@banyan.cm.nctu.edu.tw).}
\thanks{P.-N.~Chen is with the Department of Communications Engineering, National Chiao-Tung University, Taiwan, ROC (E-mail: qponing@mail.nctu.edu.tw).}
\thanks{Y.~S.~Han is with the Graduate Institute of Communication Engineering, National Taipei University, Taiwan, ROC (E-mail:yshan@mail.ntpu.edu.tw).}
\thanks{M.~H.~Kuo is currently studying toward his master degree under the
advisory of Prof.~P.-N.~Chen.}
       }

\maketitle

\begin{abstract}
The code that combines channel estimation and error protection
has received general attention recently, and has been considered
a promising methodology to compensate multi-path fading effect. It has
been shown by simulations that such code design can considerably
improve the system performance over the conventional
design with separate channel estimation and error protection
modules under the same code rate. Nevertheless, the major obstacle
that prevents from the practice of the codes is that the existing
codes are mostly searched by computers, and hence exhibit no good structure
for efficient decoding. Hence, the time-consuming exhaustive
search becomes the only decoding choice, and the
decoding complexity increases dramatically with the codeword
length. In this paper, by optimizing the signal-to-noise ratio, we
found a systematic construction for the codes for
combined channel estimation and error protection,
and confirmed its equivalence in
performance to the computer-searched codes by simulations.
Moreover, the structural codes that we construct by rules can now
be maximum-likelihoodly decodable in terms of a newly derived
recursive metric for use of the priority-first search decoding algorithm.
Thus, the decoding complexity reduces significantly when
compared with that of the exhaustive decoder. The
extension code design for fast-fading channels is also presented.
Simulations conclude that
our constructed extension code is robust in performance even
if the coherent period is shorter than the codeword length.
\end{abstract}

\section{Introduction}

The new demand of wireless communications in recent years inspires
a quick advance in wireless transmission technology. Technology
blossoms in both high-mobility low-bit-rate and low-mobility
high-bit-rate transmissions. Apparently, the next challenge in
wireless communications would be to reach high transmission rate
under high mobility. The main technology obstacle for
high-bit-rate transmission under high mobility is the seemingly
highly time-varying channel characteristic due to movement; such a
characteristic further enforces the difficulty in compensating the
intersymbol interference. Presently, a typical receiver for
wireless communications usually contains separate modules
respectively for channel estimation and channel equalization. The
former module estimates the channel parameters based on a known
training sequence or pilots, while the latter module uses these
estimated channel parameters to eliminate the channel effects due
to multipath fading. However, the effectiveness in channel fading
elimination for such a system structure may be degraded at a fast
time-varying environment, which makes high-bit-rate transmission
under high-mobility environment a big challenge.

Recent researches
\cite{Coskun}\cite{Giese}\cite{Heegard}\cite{Seshadri}\cite{Skoglund}
have confirmed that better system performance can be obtained by
jointly considering a number of system devices, such as channel
coding, channel equalization, channel estimation, and
modulation, when compared with the system with individually
optimized devices. Specially, some works on combining devices of
codeword decision and channel effect cancellation in typical
receivers can appropriately exclude channel estimation labor and still
perform well.
In 1994, Seshadri \cite{Seshadri} first proposed a blind
maximum-likelihood sequence estimator (MLSE) in which the data and
channel are simultaneously estimated. Skoglund et al
\cite{Skoglund} later provided a milestone evidence for the fact
that the joint design system is superior in combating with serious
multipath block fading. They also applied similar technique
to a multiple-input-multiple-output (MIMO) system at a subsequent work
\cite{Giese}. In short, Skoglund et al looked for the non-linear codes that
are suitable for this channel by computer search. Through
simulations, they found that the non-linear code that
combines channel estimation and error protection, when being
carefully designed by considering multipath fading effect,
outperforms a typical communication system with perfect channel
estimation by at least 2 dB. Their results suggest the high
potential of applying a single, perhaps non-linear, code to
improve the transmission rate at a highly mobile environment,
at which channel estimation becomes technically infeasible.
Similar approach was also proposed by \cite{Coskun}, and the authors
actually named such codes the \emph{training codes}. In
\cite{Chugg}, Chugg and Polydoros derived a recursive metric for
joint maximum-likelihood (ML) decoding, and hint that the
recursive metric may only be used with the sequential algorithms
\cite{Anderson}. As there are no efficient decoding
approaches for the codes mentioned above, these authors mostly
considered only codes of short length, or even just the principle
of code design for combined channel estimation and error
protection.

One of the drawbacks of these combined-channel-estimation-and-error-protection codes is that only exhaustive
search can be used to decode their codewords due to lack of
systematic structure.
Such drawback apparently inhibits the use of the codes for
combined channel estimation and error protection in practical
applications. This leads to a natural research query on how to
construct an efficiently decodable code with channel estimation and error
protection functions.

In this work, the research query was resolved by first finding that
the codeword that maximizes the system signal-to-noise ratio (SNR)
should be orthogonal to its delayed counterpart.
We then found that the code consists of the properly chosen
self-orthogonal codewords can compete with the computer-searched
codes in performance. With this self-orthogonality property,
the maximum-likelihood metrics for these structural codewords
can be equivalently fit into a recursive formula,
and hence, the priority-first search decoding algorithm can
be employed. As a consequence, the decoding complexity, as compared to the exhaustive
decoding, reduces considerably.
Extensions of our proposed coding structure that was originally
designed for channels with constant coefficients to channels
with varying channel coefficients within a codeword block
are also established. Simulations conclude that
our constructed extension code is robust even
for a channel whose coefficients vary more often than a coding block.

The paper is organized as follows.
Section \ref{SectionII} describes the system model considered,
followed by the technical backgrounds required in this work.
In Section \ref{Codeconstruction}, the coding rule that optimizes the system SNR
is established, and is subsequently used to construct
the codes for combined channel estimation
and error protection.
The corresponding recursive maximum-likelihood decoding metrics
for our rule-based systematic codes are derived in Section
\ref{SECML}. Simulations are summarized and remarked in
Section \ref{Simulationresults}.
Extension to channels with varying coefficients
within a codeword is presented in Section \ref{fastfading}.
Section \ref{Conclusions} concludes the paper.

In this work, superscripts ``$H$'' and ``$T$'' specifically
reserve for the representations of matrix Hermitian transpose and
transpose operations, respectively \cite{Har00}, and should not be
confused with the matrix exponent.

\section{Background}
\label{SectionII}

\subsection{System model and maximum-likelihood decoding criterion}\label{sectionIA}

The system model defined in this section and the notations used
throughout follow those in \cite{Skoglund}.

Transmit a codeword $\bb=[\begin{matrix}b_1,\cdots,b_N\end{matrix}]^T$, where each $b_j\in\{\pm 1\}$,
of a $(N,K)$ code $\cc$
over a block fading (specifically, quasi-static
fading) channel of memory order $(P-1)$.
Denote the channel coefficients by $\hb=[\begin{matrix}h_1,\cdots h_P\end{matrix}]^T$
that are assumed \emph{constant} within a coding block.
The complex-valued received vector is then given by:
\begin{equation}
\yb=\Bb\hb+\nb,\label{systemmodel}
\end{equation}
where $\nb$ is zero-mean complex-Gaussian distributed with
$E[\nb\nb^H]=\sigma_n^2\Ib_L$, $\Ib_L$ is the $L\times L$ identity matrix, and
$$
\Bb\triangleq\begin{bmatrix}
b_1    & 0      & \cdots & 0\\
\vdots & b_1    & \ddots & \vdots\\
b_N    & \vdots & \ddots & 0\\
0      & b_N    & \ddots & b_1\\
\vdots & \ddots & \ddots & \vdots\\
0      & 0      & \cdots & b_N
\end{bmatrix}_{L \times P}.
$$

Some assumptions are made in the following. Both the transmitter and the
receiver know nothing about the channel coefficients $\hb$, but have the
knowledge of multipath parameter $P$ or its upper bound.
Besides, there are adequate guard period between two encoding blocks
so that zero interblock interference is guaranteed.
Based on the system model in \eqref{systemmodel} and
the above assumptions, we can derive \cite{Skoglund} the least square (LS)
estimate of channel coefficients $\hb$ for a given $\bb$ (interchangeably, $\Bb$) as:
$$
\hat{\hb}=(\Bb^{T}\Bb)^{-1}\Bb^{T}\yb,
$$
and the joint maximum-likelihood (ML) decision on the transmitted codeword becomes:
\begin{equation}
\hat{\bb}=\arg\min_{\bb\in{\cc}}\|\yb-\Bb\hat{\hb}\|^2
=\arg\min_{\bb\in{\cc}}\|\yb-\Pb_B \yb\|^2 \label{MLdeci},
\end{equation}
where $\Pb_B=\Bb(\Bb^T\Bb)^{-1}\Bb^T$.
Notably, codeword $\bb$ and transformed codeword $\Pb_B$ is not one-to-one corresponding
unless the first element of $\bb$, namely $b_1$, is fixed.
For convenience, we will always set $b_1=-1$ for the codebooks we construct in the sequel.

\subsection{Summary of previous and our code designs for combined channel estimation and error protection}

In literatures, no systematic code constructions have been proposed for
combined channel estimation and error protection for
quasi-static fading channels.
Efforts were mostly placed on how to find the proper sequences to compensate
the channel fading by computer searches
\cite{Coskun}\cite{Giese}\cite{Ng}\cite{Park}\cite{Skoglund}\cite{Tellambura}\cite{Yang}.
Decodability for the perhaps structureless computer-searched codes
thus becomes an engineering challenge.

In 2003, Skoglund, Giese and Parkvall \cite{Skoglund} searched by computers for
nonlinear binary block codes suitable for combined estimation and error protection
for quasi-static fading channels by minimizing the sum of the pairwise error probabilities (PEP)
under equal prior, namely,
\begin{equation}
P_e\leq\frac 1{2^{K}}\sum_{i=1}^{2^K}\sum_{j=1,j\neq i}^{2^K}
\Pr\left(\left.\hat{\bb}=\bb(j)\right|\bb(i)\mbox{ transmitted}\right),\label{pwp}
\end{equation}
where $\bb(i)$ denotes the $i$th codeword of the $(N,K)$ nonlinear block code.
Although the operating signal-to-noise ratio (SNR) for the code search was set at $10$ dB,
their simulation results showed that the found codes perform well
in a wide range of different SNRs.
In addition, the mismatch in the relative powers of different channel coefficients,
as well as in the channel Rice factors \cite{TAG03}, has little effect on the resultant performance.
It was concluded that in comparison with
the system with the benchmark error correcting code and the perfect channel estimator,
significant performance improvement can be obtained
by adopting their computer-searched nonlinear codes.

Later in 2005, Coskun and Chugg \cite{Coskun} replaced the PEP in \eqref{pwp}
by a properly defined pairwise distance measure between two codewords, and proposed
a suboptimal greedy algorithm to speed up the code search process.
In 2007, Giese and Skoglund \cite{Giese} re-applied their original idea
to the single- and multiple-antenna systems, and
used the asymptotic PEP and the generic gradient-search algorithm
in place of the PEP and the simulated annealing algorithm in \cite{Skoglund}
to reduce the system complexity.

At the end of \cite{Skoglund}, the authors pointed out that ``an important
topic for further research is to study how the decoding complexity of the
proposed scheme can be decreased.'' They proceeded to state that along this
research line, ``one main issue is to investigate what kind of structure should
be enforced on the code to allow for simplified decoding.''

Stimulating from
these ending statements, we take a different approach for code design. Specifically, we
pursued and established a systematic code design rule for combined channel
estimation and error protection for quasi-static fading channels, and confirmed
that the codes constructed based on such rule maximize the average
system SNR. As so happened that the computer-searched code
in \cite{Skoglund} satisfies such rule, its insensitivity to SNRs, as
well as channel mismatch, somehow finds the theoretical footing. Enforced by
the systematic structure of our rule-based constructed codes, we can then
derive a recursive {\em maximum-likelihood} decoding metric for use of
priority-first search decoding algorithm. The decoding complexity is therefore
significantly decreased at moderate-to-high SNRs as contrary to the obliged
exhaustive decoder for the structureless computer-searched codes.

It is worth mentioning that although the codes searched by computers in \cite{Giese}\cite{Skoglund}
 target the unknown channels,
for which the channel coefficients are assumed constant in a coding block, the evaluation of
the PEP criterion does require to presume the knowledge of channel statistics.
The code constructed based on the rule we proposed, however, is guaranteed
to maximize the system SNR regardless of the statistics of
the channels. This hints that our code can still be well applied to the
situation where channel blindness becomes a strict system restriction.
Details will be introduced in subsequent sections.

\subsection{Maximum-likelihood priority-first search decoding algorithm}
\label{algorithm}

For a better understanding,
we give a short description of a code tree for the $(N,K)$ code $\mathcal{C}$
over which the decoding search is performed
before our describing the priority-first search decoding algorithm in this subsection.

A code tree of a $(N,K)$ binary code represents
every codeword as a path on a binary tree as shown in Fig.~\ref{codetree}.
The code tree consists of $(N+1)$ levels. The single leftmost node at level zero is usually
called the {\it origin node}. There are at most two branches
leaving each node at each level.  The $2^{K}$ rightmost nodes at level $N$ are called the
{\it terminal nodes}.

Each branch on the code tree is labeled with the appropriate
code bit $b_i$. As a convention, the path from the single origin node to one of the $2^K$ terminal nodes
is termed the {\it code path} corresponding to the
codeword. Since there is a one-to-one correspondence between the codeword and the code path of
$\mathcal{C}$, a codeword can be interchangeably referred to by its respective code path
or the branch labels that the code path traverses.
Similarly, for any node in the code tree,
there exists a unique path traversing from the single original node
to it; hence, a node can also be interchangeably indicated by the path (or the path labels) ending at it.
We can then denote the path ending at a node at level $\ell$
by the branch labels $[b_1,b_2,\cdots,b_\ell]$ it traverses.
For convenience, we abbreviates $[b_1,b_2,\cdots,b_\ell]^T$ as $\bb_{(\ell)}$,
and will drop the subscript when $\ell=N$.
The successor pathes of a path $\bb_{(\ell)}$
are those whose first $\ell$ labels are exactly the same as $\bb_{(\ell)}$.

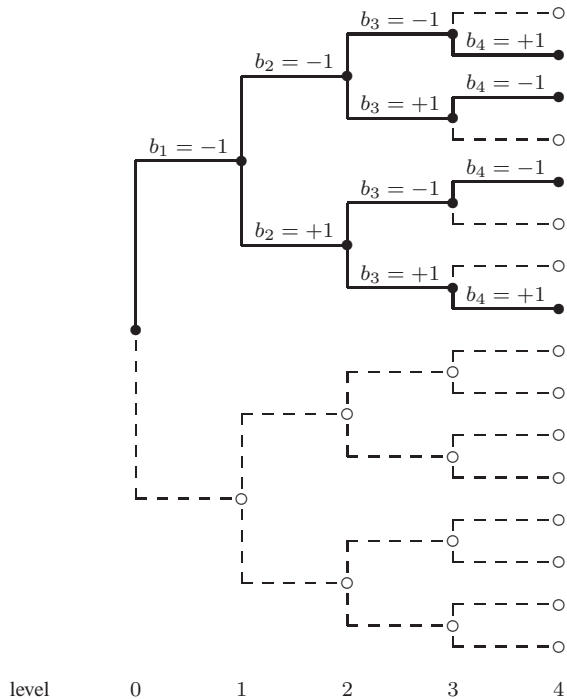
\begin{figure}[hbt]
\begin{center}
\setlength{\unitlength}{0.8pt}
\begin{picture}(250,330)(40,0)
\multiput(200,320)(10,0){5}{\line(1,0){5}}
\multiput(200,260)(10,0){5}{\line(1,0){5}}
\multiput(200,220)(10,0){5}{\line(1,0){5}}
\multiput(200,200)(10,0){5}{\line(1,0){5}}
\multiput(200,160)(10,0){5}{\line(1,0){5}}
\multiput(200,140)(10,0){5}{\line(1,0){5}}
\multiput(200,120)(10,0){5}{\line(1,0){5}}
\multiput(200,100)(10,0){5}{\line(1,0){5}}
\multiput(200,80)(10,0){5}{\line(1,0){5}}
\multiput(200,60)(10,0){5}{\line(1,0){5}}
\multiput(200,40)(10,0){5}{\line(1,0){5}}
\multiput(200,20)(10,0){5}{\line(1,0){5}}
\put(250,320){\circle{5}}
\put(250,300){\circle*{5}}
\put(250,280){\circle*{5}}
\put(250,260){\circle{5}}
\put(250,240){\circle*{5}}
\put(250,220){\circle{5}}
\put(250,200){\circle{5}}
\put(250,180){\circle*{5}}
\multiput(250,20)(0,20){8}{\circle{5}}
\put(200,315){\line(0,1){5}}
\put(200,260){\line(0,1){5}}
\put(200,220){\line(0,1){5}}
\put(200,195){\line(0,1){5}}
\multiput(200,20)(0,40){4}{\line(0,1){5}}
\multiput(200,35)(0,40){4}{\line(0,1){5}}
\multiput(150,150)(10,0){5}{\line(1,0){5}}
\multiput(150,110)(10,0){5}{\line(1,0){5}}
\multiput(150,70)(10,0){5}{\line(1,0){5}}
\multiput(150,30)(10,0){5}{\line(1,0){5}}
\multiput(200,190)(0,40){4}{\circle*{5}}
\multiput(200,30)(0,40){4}{\circle{5}}
\multiput(150,110)(0,10){2}{\line(0,1){5}}
\multiput(150,30)(0,10){2}{\line(0,1){5}}
\multiput(150,135)(0,10){2}{\line(0,1){5}}
\multiput(150,55)(0,10){2}{\line(0,1){5}}
\multiput(100,130)(10,0){5}{\line(1,0){5}}
\multiput(100,50)(10,0){5}{\line(1,0){5}}
\multiput(150,210)(0,80){2}{\circle*{5}}
\multiput(150,50)(0,80){2}{\circle{5}}
\multiput(100,50)(0,10){4}{\line(0,1){5}}
\multiput(100,95)(0,10){4}{\line(0,1){5}}
\multiput(50,90)(10,0){5}{\line(1,0){5}}
\put(100,90){\circle{5}}
\put(100,250){\circle*{5}}
\multiput(50,165)(0,-10){8}{\line(0,-1){5}}
\put(50,170){\circle*{5}}

\scriptsize{
\put(0,0){\makebox(0,0){level}}
\put(50,0){\makebox(0,0){$0$}}
\put(100,0){\makebox(0,0){$1$}}
\put(150,0){\makebox(0,0){$2$}}
\put(200,0){\makebox(0,0){$3$}}
\put(250,0){\makebox(0,0){$4$}}

\put(50,250){\makebox(50,13){$b_1=-1$}}

\put(100,210){\makebox(50,13){$b_2=+1$}}
\put(100,290){\makebox(50,13){$b_2=-1$}}

\put(150,190){\makebox(50,13){$b_3=+1$}}
\put(150,230){\makebox(50,13){$b_3=-1$}}
\put(150,270){\makebox(50,13){$b_3=+1$}}
\put(150,310){\makebox(50,13){$b_3=-1$}}

\put(200,180){\makebox(50,13){$b_4=+1$}}
\put(200,240){\makebox(50,13){$b_4=-1$}}
\put(200,280){\makebox(50,13){$b_4=-1$}}
\put(200,300){\makebox(50,13){$b_4=+1$}}
}

\thicklines
\put(200,300){\line(1,0){50}}
\put(200,280){\line(1,0){50}}
\put(200,240){\line(1,0){50}}
\put(200,180){\line(1,0){50}}

\put(200,300){\line(0,1){10}}
\put(200,270){\line(0,1){10}}
\put(200,230){\line(0,1){10}}
\put(200,180){\line(0,1){10}}

\multiput(150,190)(0,40){4}{\line(1,0){50}}
\multiput(150,190)(0,80){2}{\line(0,1){40}}
\multiput(100,210)(0,80){2}{\line(1,0){50}}
\put(100,210){\line(0,1){80}}
\put(50,250){\line(1,0){50}}
\put(50,170){\line(0,1){80}}
\end{picture}
\caption{The code tree for a computer-searched PEP-minimum $(4,2)$ code with $b_1$ fixed as $-1$.}\label{codetree}
\end{center}
\end{figure}

The priority-first search on a code tree
is guided by an evaluation function $f$ that is defined for every path.
It can be typically algorithmized as follows.

\bigskip

\begin{list}{Step~\arabic{step}.}
    {\usecounter{step}
    \setlength{\labelwidth}{1cm}
    \setlength{\leftmargin}{1.6cm}\slshape}
\item {\bf(Initialization)} Load the Stack with the path that ends at the original node.

\item \label{step} {\bf(Evaluation)} Evaluate the $f$-function values of the successor paths of the current top path in the
      Stack, and delete this top path from the Stack.

\item {\bf (Sorting)} Insert the successor paths obtained in Step~\ref{step} into the  Stack such
      that the paths in the Stack are ordered according to ascending $f$-function values of them.

\item {\bf (Loop)} If the top path in the Stack ends at a terminal node in the
      code tree, output the labels corresponding
      to the top path, and the algorithm stops; otherwise, go to Step~\ref{step}.
\end{list}

It remains to find the evaluation function $f$ that secures the maximum-likelihoodness
of the output codeword.
We begin with the introduction of a sufficient condition under which
the above priority-first search algorithm
guarantees to locate the code path with the smallest $f$-function value
among all code paths of $\cc$.

\bigskip

\begin{lemma}\label{lemma0} If $f$ is non-decreasing along every path $\bb_{(\ell)}$ in the code tree, i.e.,
\begin{equation}
\label{optimalcon}
f\left(\bb_{(\ell)}\right)\leq \min_{\left\{\tilde\bb\in\cc~:~\tilde\bb_{(\ell)}=\bb_{(\ell)}\right\}}f(\tilde\bb),
\end{equation}
the priority-first search algorithm always outputs
the code path with the smallest $f$-function value among all code paths of $\cc$.
\end{lemma}
\begin{proof}
Let $\bb^\ast$ be the first top path that
reaches a terminal node (and hence, is the output code path
of the priority-first search algorithm.)
Then, {\em Step 3} of the algorithm ensures that
$f\left(\bb^\ast\right)$ is no larger than
the $f$-function value of any path currently in the Stack.
Since condition \eqref{optimalcon} guarantees that the $f$-function value
of any other code path, which should be the offspring of some path $\bb_{(\ell)}$
existing in the Stack, is no less than $f\left(\bb_{(\ell)}\right)$,
we have
$$f\left(\bb^\ast\right)\leq f\left(\bb_{(\ell)}\right)\leq
\min_{\left\{\tilde\bb\in\cc~:~\tilde\bb_{(\ell)}=\bb_{(\ell)}\right\}}f(\tilde\bb).$$
Consequently, the lemma follows.
\end{proof}

\bigskip

In the design of the search-guiding function $f$,
it is convenient to divide it into the sum of two parts. In order to perform maximum-likelihood decoding,
the first part $g$ can be directly defined based on the maximum-likelihood
metric of the codewords such that from \eqref{MLdeci},
$$
\arg\min_{\bb\in\cc}g(\bb)=\arg\min_{\bb\in{\cc}}\|\yb-\Pb_B \yb\|^2.
$$
After $g$ is defined, the second part $h$ can be designed to
validate \eqref{optimalcon} with $h(\bb)=0$ for any $\bb\in\cc$.
Then, from $f(\bb)=g(\bb)+h(\bb)=g(\bb)$ for all $\bb\in\cc$,
the desired maximum-likelihood priority-first search decoding
algorithm is established since \eqref{optimalcon} is valid.

In principle, both $g(\cdot)$
and $h(\cdot)$ range over all possible paths in the code tree.
The first part, $g(\cdot)$, is simply a function of all the branches
traversed thus far by the path, while the second part, $h(\cdot)$, called the
{\it heuristic function}, helps predicting
a future route from the end node of the current path to a
terminal node \cite{Han}. Notably, the design of the heuristic function $h$
that makes valid condition \eqref{optimalcon} is not unique. Different designs may result in variations in computational complexity.

We close this section by summarizing the target of this work based on what
have been mentioned in this section.
\begin{enumerate}
\item \label{item1} A code of comparable performance to the computer-searched code
is constructed according to certain rules so that its code tree can be efficiently and
systematically generated (Section~\ref{Codeconstruction}).

\item \label{item2} Efficient recursive computation of the maximum-likelihood evaluation function $f$
from the predecessor path to the successor paths is established (Section~\ref{SECML}).

\item With the availability of items \ref{item1} and \ref{item2}, the construction and maximum-likelihood decoding
of codes with longer codeword length becomes possible, and hence, makes the assumption
that the unknown channel coefficients $\hb$ are fixed during a long coding block somewhat impractical
especially for mobile transceivers. Extension of items \ref{item1} and \ref{item2}
to the unknown channels whose channel coefficients may change several times during one coding block
will be further proposed (Section~\ref{fastfading}).
\end{enumerate}

\section{Code Construction}
\label{Codeconstruction}

In this section,
the code design rule that guarantees the maximization of
the system SNR regardless of the channel
statistics is presented, followed by the algorithm to generate
the code based on such rule.

\subsection{Code rule that maximizes the average SNR}
\label{IIIA}

A known inequality \cite{Patel} for the multiplication of two positive semidefinite
Hermitian matrices, $\Ab$ and $\Bb$, is that
\begin{eqnarray}
\tr(\Ab\Bb)\leq \tr(\Ab)\cdot \lambda_{\max}(\Bb),\label{ineq1}
\end{eqnarray}
where $\tr(\cdot)$ represents the matrix trace operation, and
$\lambda_{\max}(\Bb)$ is the maximal eigenvalue of
$\Bb$ \cite{Har00}. The above inequality holds with equality when $\Bb$ is an identity matrix.

From the system model $\yb=\Bb\hb+\nb$, it can be derived that
the average SNR satisfies:
\begin{eqnarray}
\mbox{Average SNR}&=&\frac{E[\|\Bb\hb\|^2]}{E[\|\nb\|^2]}\nonumber\\
&=&\frac{E[\tr(\hb^H\Bb^T\Bb\hb)]}{L\sigma_n^2}\nonumber\\
&=&\frac{\tr(E[\hb\hb^H]\Bb^T\Bb)}{L\sigma_n^2}\nonumber\\
&=&\frac{N}{L}\frac 1{\sigma_n^2}\tr\left(E[\hb\hb^H]\frac{1}{N}\Bb^T\Bb\right)\nonumber\\
&\leq&\frac{N}{L}\frac 1{\sigma_n^2}\tr(E[\hb\hb^H])\lambda_{\max}\left(\frac{1}{N}\Bb^T\Bb\right).\nonumber
\end{eqnarray}
Then, the theories on Ineq.~\eqref{ineq1} result
that taking
\begin{equation}
\label{coderule}
\frac 1N\Bb^T\Bb=\Ib_P\triangleq\begin{bmatrix}
1&0&\cdots&0\\
0&1&\cdots&0\\
\vdots&\vdots&\ddots&\vdots\\
0&0&\cdots&1
\end{bmatrix}_{P\times P}
\end{equation}
will optimize the average SNR regardless
of the statistics of $\hb$ \cite{Ganesan}.

Existence of codeword sequences satisfying \eqref{coderule}
is promised only for $P=2$ with $N$ odd (and trivially, $P=1$).
In some other cases such as $P=3$, one can only design codes to approximately satisfy \eqref{coderule}
as:
$$
\frac 1N\Bb^T\Bb=\begin{bmatrix}
1     & \pm\dis\frac{1}N & 0     \\
\pm\dis\frac{1}N & 1     & \pm\dis\frac{1}N \\
0     & \pm\dis\frac{1}N & 1
\end{bmatrix}\mbox{ for }N\mbox{ even,}$$
and
$$\frac 1N\Bb^T\Bb=
\begin{bmatrix}
1     & 0 & \pm\dis\frac{1}N      \\
0& 1     &0 \\
 \pm\dis\frac{1}N&0 & 1
\end{bmatrix}\mbox{ for }N\mbox{ odd}.
$$
Owing to this observation, we will relax \eqref{coderule}
to allow some off-diagonal entries
in $\Bb^T\Bb$ to be either $1$ or $-1$ whenever a strict maintenance of \eqref{coderule}
is impossible.

Empirical examination by simulated-annealing code-search algorithm
shows that for $4\leq N\leq 16$ and $N$ even,
the best half-rate codes that minimize the sum of PEPs in \eqref{pwp}
under\footnote{The adopted statistical parameters of $\hb$
follow those in \cite{Skoglund}.} complex zero-mean Gaussian distributed $\hb$ with $E[\hb\hb^H]=(1/2)\Ib_P$ and $P=2$
all satisfy that
\begin{equation}
\label{coderule1}
\Bb^T\Bb=\begin{bmatrix}N&\pm 1\\\pm 1&N\end{bmatrix}
\end{equation}
except three codewords at $N=14$.
A possible cause for the appearance of three exception
codewords at $N=14$ is that the best code that minimizes
the sum of the pairwise error probabilities
may not be the truly optimal code that
reaches the smallest error probability, and hence, does not necessarily
yield the maximum average SNR as demanded by \eqref{coderule}. We have also
obtained and examined the computer-searched code used in \cite{Skoglund} for $N=22$,
and found as anticipated that every codeword carries the property
of \eqref{coderule1}.

The operational meaning of the condition $\Bb^T\Bb=N\cdot\Ib_P$
is that the codeword is orthogonal to its shifted counterpart,
and hence, a space-diversity nature is implicitly enforced.
This coincides with the conclusion made in \cite{Crozier}
that the training sequence satisfying that $\Bb^T\Bb$ is proportional
to $\Ib_P$ can provide optimal channel estimation performance.
It should be mentioned that codeword condition \eqref{coderule}
has been identified in \cite{Giese}, and the authors in \cite[pp.~1591]{Giese} remarked
that a code sequence with certain aperiodic autocorrelation property
can possibly be exploited in future code design approaches, which is
one of the main research goals of this paper.

\subsection{Equivalent system model for combined channel estimation and error protection
codes}

\begin{figure}[tbp]
\setlength{\unitlength}{0.8pt}
\begin{picture}(480,150)(0,50)
\put(40,82){\makebox(110,0){\fcolorbox[rgb]{0.9,1,0.9}{0.9,1,0.9}{$~~~~~~~~~~~~~~~~~~~~$}}}
\put(40,90){\makebox(110,0){\fcolorbox[rgb]{0.9,1,0.9}{0.9,1,0.9}{$~~~~~~~~~~~~~~~\nb~~$}}}
\put(40,100){\makebox(110,0){\fcolorbox[rgb]{0.9,1,0.9}{0.9,1,0.9}{$~~~~~~~~~~~~~~~~~~~~$}}}
\put(40,107){\makebox(110,0){\fcolorbox[rgb]{0.9,1,0.9}{0.9,1,0.9}{$~~~~~~~~~~~~~~~~~~~~$}}}
\put(40,114){\makebox(110,0){\fcolorbox[rgb]{0.9,1,0.9}{0.9,1,0.9}{$~~~~~~~~~~~~~~~~~~~~$}}}
\put(40,121){\makebox(110,0){\fcolorbox[rgb]{0.9,1,0.9}{0.9,1,0.9}{$~~~~~~~~~~~~~~~~~~~~$}}}
\put(40,128){\makebox(110,0){\fcolorbox[rgb]{0.9,1,0.9}{0.9,1,0.9}{$~~~~~~~~~~~~~~~~~~~~$}}}
%\put(40,130){\makebox(110,0){\fcolorbox[rgb]{0.9,1,0.9}{0.9,1,0.9}{$~~~~~~~~~~~~~~~~~~~~$}}}
\put(40,140){\makebox(110,0){\fcolorbox[rgb]{0.9,1,0.9}{0.9,1,0.9}{~~~~$\hb$~~~~~~~~~$\oplus$~~}}}
\put(40,152){\makebox(110,0){\fcolorbox[rgb]{0.9,1,0.9}{0.9,1,0.9}{$~~~~~~~~~~~~~~~~~~~~$}}}
\put(40,153){\makebox(110,0){\fcolorbox[rgb]{0.9,1,0.9}{0.9,1,0.9}{$~~~~~~~~~~~~~~~~~~~~$}}}
\put(40,160){\makebox(110,0){\fcolorbox[rgb]{0.9,1,0.9}{0.9,1,0.9}{$~~~~~~~~~~~~~~~~~~~~$}}}
\put(40,167){\makebox(110,0){\fcolorbox[rgb]{0.9,1,0.9}{0.9,1,0.9}{$~~~~~~~~~~~~~~~~~~~~$}}}
\put(40,174){\makebox(110,0){\fcolorbox[rgb]{0.9,1,0.9}{0.9,1,0.9}{$~~~~~~~~~~~~~~~~~~~~$}}}

%\thicklines
\put(0,140){\vector(1,0){50}} \put(50,120){\framebox(40,40)}
\put(90,140){\vector(1,0){33}} \put(129,100){\vector(0,1){31}}
\put(135,140){\vector(1,0){45}}
\put(180,110){\framebox(120,60){\shortstack{Outer
product\\demodulator}}} \put(300,140){\vector(1,0){50}}
\put(350,110){\framebox(150,60){\shortstack{Minimum
Euclidean\\Distance Selector}}} \put(500,140){\vector(1,0){50}}
\put(40,80){\dashbox(110,98)} \put(30,55){\dashbox(280,135)}
\put(30,55){\makebox(280,25){\bf Equivalent Channel}}

\put(0,145){\makebox(0,0)[b]{$\bb$}}
\put(310,145){\makebox(40,0)[b]{$\yb\yb^H$}}
\put(520,145){\makebox(30,0)[b]{$\hat\bb(\yb\yb^H)$}}
\end{picture}
\caption{Equivalent system model for combined channel estimation
and error protection codes. } \label{fig:equivalent_channel}
\end{figure}
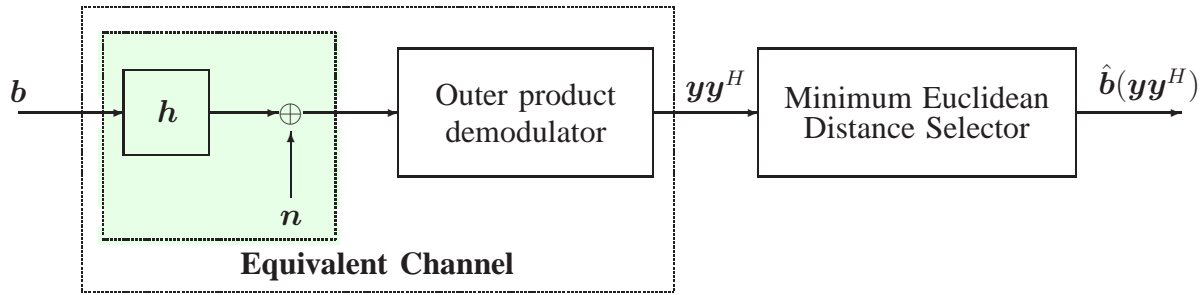

By noting\footnote{
The validity of the claimed statement here does not require the SNR-optimization condition $\Bb^T\Bb=N\Ib_P$.}
that $\Pb_B$ is idempotent and symmetric,
and both $\tr(\Pb_B)$ and $\|\vectorize(\Pb_B)\|^2$ are equal to $P$,
where $\vectorize(\cdot)$ denotes the operation to transform an
$(M\times N)$ matrix into a $(MN\times 1)$ vector,\footnote{
$\vectorize(\Ab)$ for a matrix $\Ab$ is defined as:
$$
\vectorize(\Ab)=\vectorize\left(\begin{bmatrix}
a_{1,1}& a_{1,2}& \cdots& a_{1,S}\\
\vdots & \vdots & \cdots& \vdots \\
a_{R,1}& a_{R,2}& \cdots& a_{R,S}
\end{bmatrix}\right)=\begin{bmatrix}a_{1,1}&\cdots&a_{R,1}&a_{1,2}&\cdots&a_{R,S}\end{bmatrix}^T.
$$}
the best joint maximum-likelihood decision in \eqref{MLdeci}
can be reformulated as:
\begin{eqnarray}
\hat{\bb}&=&\arg\min_{\bb\in{\mathcal{C}}}\|\yb-\Pb_B \yb\|^2\nonumber\\
&=&\arg\min_{\bb\in{\mathcal{C}}}(\yb-\Pb_B \yb)^H(\yb-\Pb_B
\yb)\nonumber\\
&=&\arg\min_{\bb\in{\mathcal{C}}}(\yb^H\yb-\yb^H\Pb_B
\yb)\nonumber\\
&=&\arg\min_{\bb\in{\mathcal{C}}}-\tr(\Pb_B\yb\yb^H)\label{PByyH}\\
&=&\arg\min_{\bb\in{\mathcal{C}}}
\left(\|\vectorize(\yb\yb^H)\|^2-\vectorize(\Pb_B)^T\vectorize(\yb\yb^H)-\vectorize(\yb\yb^H)^H\vectorize(\Pb_B)+\|\vectorize(\Pb_B)\|^2\right)\nonumber\\
&=&\arg\min_{\bb\in{\mathcal{C}}}\|\vectorize(\yb\yb^H)-\vectorize(\Pb_B)\|^2.\label{yy_Pb}
\end{eqnarray}
We therefore transform the original system in \eqref{systemmodel}
to an equivalent system model that contains an \emph{outer product
demodulator} and a \emph{minimum Euclidean distance selector}
at the $\Pb_B$-domain as shown in
Fig.~\ref{fig:equivalent_channel}. As the outer product
demodulator can be viewed as a generalization of the
\emph{square-law combining} that is of popular use in non-coherent
detection for both slow and fast fading \cite{Proakis}, the above
equivalent transformation suggests a potential application of
combined channel estimate and error protection codes for the
non-coherent system in which the fading is rapid enough to
preclude a good estimate of the channel coefficients. Further
discussion on how to design codes for unknown fast-fading channels
will be continued in Section~\ref{fastfading}.

As a consequence of \eqref{yy_Pb}, the maximum-likelihood decoding is to find the codeword
$\Pb_{B}$ whose
Euclidean distance to $\yb\yb^H$ is the smallest.
Similar to \eqref{pwp}, we can then bound the error probability by:
\begin{eqnarray}
&&\!\!\!\!\!\!\!\!\!P_e\leq\frac 1{2^{K}}\sum_{i=1}^{2^K}\sum_{j=1\atop j\neq i}^{2^K}
\Pr\!\left(\left.\|\vectorize(\yb\yb^H)-\vectorize(\Pb_{B(j)})\|^2
\!<\!\|\vectorize(\yb\yb^H)-\vectorize(\Pb_{B(i)})\|^2\right|\bb(i)\mbox{ transmitted}\right).\nonumber\\\label{PEPbound}
\end{eqnarray}
The PEP-based upper bound in \eqref{PEPbound} hints that
a good code design should have an adequately large pairwise Euclidean distance
\begin{equation}
\label{pwdis}
\|\vectorize(\Pb_{B(i)})-\vectorize(\Pb_{B(j)})\|^2
\end{equation}
among all codeword pairs $\Pb_{B(i)}$ and $\Pb_{B(j)}$, where $\Pb_{B(i)}$ is
the equivalent codeword at the $\Pb_B$-domain, and is one-to-one corresponding
to the codeword $\bb(i)$ if $b_1$ is fixed and known. Based on this
observation, we may infer under equal prior that a uniform draw of codewords
satisfying $\Bb^T\Bb=N\cdot\Ib_P$ at the $\Pb_B$-domain may asymptotically result
in a good code. In light of the one-to-one
correspondence relation between $\bb$ and $\Pb_B$, we may further infer that
uniform selection of codewords in the set of
$$\ac\triangleq\left\{\Pb_B=\Bb(\Bb^T\Bb)^{-1}\Bb~:~\exists \bb\in\{\pm 1\}^N \mbox{ such that
}\Bb^T\Bb=N\cdot\Ib_P\right\}$$
is conceptually equivalent to uniform-pick of codewords in
$\{\bb\in\{\pm 1\}^N:\Bb^T\Bb=N\Ib_P\}$.

Recall that in order to perform the priority-first search decoding on a code tree,
an efficient and systematic way to generate the code tree (or more specifically,
an efficient and systematic way
to generate the successor paths of the top path) is necessary.
The uniform pick principle then suggests that considering only the codewords with the same prefix $[b_1,\cdots,b_\ell]$,
the ratio of the number of codewords satisfying $b_{\ell+1}=-1$
with respect to the candidate sequence pool
shall be made equal to
that of codewords satisfying $b_{\ell+1}=1$, whenever possible.
This can be mathematical interpreted as:
\begin{equation}
\label{ratio1}
\frac{|\cc(b_1,b_2,\cdots,b_\ell,+1)|}{\left|\ac(b_1,b_2,\cdots,b_\ell,+1|N\cdot\Ib_P)\right|}
\approx\frac{|\cc(b_1,b_2,\cdots,b_\ell,-1)|}{\left|\ac(b_1,b_2,\cdots,b_\ell,-1|N\cdot\Ib_P)\right|},
\end{equation}
where
$\cc(\bb_{(\ell)})$ is the set of all codewords
whose first $\ell$ bits equal $b_1,b_2,\cdots,b_\ell$,
and
$\ac(\bb_{(\ell)}|\Gb)$ is the set of all possible
$\pm 1$-sequences of length $N$, whose first $\ell$ bits equal
$b_1,b_2,\cdots,b_\ell$ and whose $\Bb$-representation satisfies
$\Bb^T\Bb=\Gb$.
Accordingly, if $|\ac(\bb_{(\ell)}|\Gb)|$ can be computed explicitly,
the desired efficient and systematic generation of the code tree becomes
straightforward.

Simulations on the above uniform selective code over $\{\bb\in\{\pm 1\}^N:\Bb^T\Bb=N\Ib_P\}$
show that its performance is almost the same as the computer-searched
code that minimizes the sum of PEPs.
Hence, the maximizing-the-pairwise-Euclidean-distance intuition we adopt for code design
performs well as we have anticipated.

In the next subsection, we will provide an exemplified encoding algorithm based on
the above basic rule specifically for channels of memory order $1$, namely, $P=2$.
The encoding algorithm for
larger memory order can be similarly built.

\subsection{Exemplified encoding algorithm for channels of memory order one}\label{encoding}
\label{Encode-Subsection}

Before the presentation
of the exemplified encoding algorithm, the explicit formula for $|\ac(\bb_{(\ell)}|\Gb)|$
needs to be established first.\footnote{ $|\ac(\bb_{(\ell)}|\Gb)|$ may not
have an explicit close-form formula for memory order higher than one. However, our
encoding algorithm can still be applied as long as
$|\ac(\bb_{(\ell)}|\Gb)|$ can be pre-calculated (cf.~Appendix). }

\bigskip

\begin{lemma}\label{Lem1}
Fix $P=2$.
Then, for $N$ odd, and $\Gb=N\cdot\Ib_P$,
\begin{eqnarray*}
|\ac(\bb_{(\ell)}|\Gb)|=
  \dis\binom{N-\ell}{(N-\ell-m_\ell)/2}
\oneb\left\{\left|m_\ell\right|\leq N-\ell\right\}\mbox{ for }1\leq\ell\leq N,
\end{eqnarray*}
where
$\oneb\{\cdot\}$ is the set indicator function, and
$m_\ell\triangleq(b_1b_2+\cdots+b_{\ell-1}b_{\ell})\cdot\oneb\{\ell>1\}$.

In addition, for $N$ even, and
$$\Gb_1\triangleq\begin{bmatrix}N&-1\\-1&N\end{bmatrix}
\quad\mbox{and}\quad \Gb_2\triangleq\begin{bmatrix}N&1\\1&N\end{bmatrix},$$
\begin{eqnarray*}
|\ac(\bb_{(\ell)}|\Gb_\theta)|=
  \dis\binom{N-\ell}{(N-\ell+(-1)^\theta-m_\ell)/2}
\oneb\left\{\left|(-1)^\theta-m_\ell\right|\leq N-\ell\right\}\mbox{ for }1\leq\ell\leq N.
\end{eqnarray*}
Here, we assume that $\binom{0}{0}=1$ specifically for the case of $\ell=N$.
\end{lemma}

\begin{proof}
The lemma requires
\begin{equation}
c=b_1b_2+b_2b_3+\cdots+b_{N-1}b_N=m_\ell+b_\ell b_{\ell+1}+\cdots+b_{N-1}b_N,\label{lemma1}
\end{equation}
where $c=0,-1,+1$ respectively for $\Gb$, $\Gb_1$ and $\Gb_2$.
In order to satisfy \eqref{lemma1}, there should be
$(N-\ell+c-m_\ell)/2$ of $\{b_\ell b_{\ell+1},b_{\ell+1}b_{\ell+2},\cdots,b_{N-1}b_N\}$
equal to $1$, and the remaining of them equal $-1$, provided that
$(N-\ell+c-m_\ell)/2\geq 0$ and $(N-\ell)-(N-\ell+c-m_\ell)/2\geq 0$.
Notably, $(N-\ell+c-m_\ell)$ is always an even number for the cases considered in the lemma.
The proof is then completed by the observation that $[b_\ell b_{\ell+1},b_{\ell+1}b_{\ell+2},\cdots,b_{N-1}b_N]$ and
$[b_{\ell+1},b_{\ell+2},\cdots,b_N]$ are one-to-one correspondence
for given $b_\ell$.
\end{proof}

\bigskip

It is already hint in the above lemma that for $N$ odd, our encoding algorithm will uniformly pick
$2^K$ codewords
from the candidate sequences satisfying the exact SNR-maximization condition $\Bb^T\Bb=N\cdot\Ib_P$.
However, for $N$ even, two conditions on candidate sequences will be used.
Half of the codewords will be uniformly drawn
from those candidate sequences satisfying $\Bb^T\Bb=\Gb_1$, and
the other half of the codewords agree with $\Bb^T\Bb=\Gb_2$.
The proposed codeword selection process is simply to list all the candidate sequences
in binary-alphabetical order, starting from zero,
and uniformly pick the codewords
from the ordered list in every $\Delta$ interval, where
$$\Delta=\left\lfloor\frac{|\ac(b_1=-1|\Hb)|-1}{2^{K}/\Theta-1}\right\rfloor,$$
where $\Hb$ represents the desired $\Bb^T\Bb$,
and $\Theta$ is the number of conditions and equals $1$ for $N$ odd,
and $2$ for $N$ even.
As a result, the selected codewords are those sequences
with index $i\times\Delta$ for integer $i$.
The encoding algorithm is summarized in the following.

\bigskip

\begin{list}{Step~\arabic{step}.}
    {\usecounter{step}
    \setlength{\labelwidth}{1cm}
    \setlength{\leftmargin}{1.6cm}\slshape}
\item\label{Step1} {\bf (Input)} Let $i$ be the index of the requested codeword in
the desired $(N,K)$ block code, where $0\leq i\leq 2^K-1$.

\item {\bf (Initialization)}
Set $\Theta=1$ for $N$ odd, and $2$ for $N$ even. Let $b_1=-1$.
Put (in terms of the notations in Lemma~\ref{Lem1}):
$$
\Hb=\left\{\begin{array}{ll}
\Gb,&\mbox{if }N\mbox{ odd};\\
\Gb_1,&\mbox{if }N\mbox{ even and }0\leq i<2^{K-1};\\
\Gb_2,&\mbox{if }N\mbox{ even and }2^{K-1}\leq i<2^K.
\end{array}\right.
$$
Compute
$$\Delta=\left\lfloor\frac{|\ac(b_1|\Hb)|-1}{2^{K}/\Theta-1}\right\rfloor.$$
Also, re-adjust $i=i-2^{K-1}$ if $N$ is even and $2^{K-1}\leq i<2^K$.

Let the minimum candidate sequence index $\rho_{\min}$ and the maximum candidate sequence index
$\rho_{\max}$
in $\ac(b_1|\Hb)$ be respectively
$$\rho_{\min}=0\quad{\rm and}\quad
\rho_{\max}=|\ac(b_1|\Hb)|-1.$$

Initialize $\ell=1$ and $\rho=i\times\Delta$.

\item {\bf (Generation of the next code bit)}\label{Step3}

Set $\ell=\ell+1$, and compute
$\gamma=|\ac(\bb_{(\ell-1)},-1|\Hb)|$.

If $\rho<\rho_{\min}+\gamma$, then the next code bit $b_{\ell}=-1$, and
re-adjust $\rho_{\max}=\rho_{\min}+\gamma-1$;

else, the next code bit $b_{\ell}=1$, and re-adjust $\rho_{\min}=\rho_{\min}+\gamma$.

\item {\bf (Loop)}
If $\ell=N$, output codeword $\bb$, and the algorithm stops; otherwise, go to Step \ref{Step3}.
\end{list}

\section{Maximum-Likelihood Priority-First Search Decoding of Combined Channel Estimation
and Error Protection Codes}
\label{SECML}

In this section, a recursive maximum-likelihood metric $g$
and its heuristic function $h$
for use of the
priority-first search decoding algorithm to decode
the structural codewords over multiple code trees
are established.

\subsection{Recursive maximum-likelihood metric $g$ for priority-first search over multiple code trees}
\label{Rmmg}

Let $\cc_\theta$ be the set of the codewords that satisfy
$\Bb^T\Bb=\Gb_\theta$, where $1\leq\theta\leq\Theta$, and
assume that $\cc=\cup_{1\leq\theta\leq\Theta}\cc_\theta$, and
$\cc_\theta\cap\cc_\eta=\emptyset$
whenever $\theta\neq\eta$.
Then, by denoting for convenience $\Db_\theta=\Gb_\theta^{-1}$,
we can continue the derivation of the maximum-likelihood criterion from \eqref{PByyH} as:
\begin{eqnarray}
\hat{\bb}&=&\arg\min_{\bb\in\cc}\sum_{\theta=1}^\Theta
\left[-\tr(\Bb\Db_\theta\Bb^T\yb\yb^H)\right]\bf{1}\{\bb\in\cc_\theta\}\nonumber\\
&=&\arg\min_{\bb\in{\mathcal{C}}}\sum_{\theta=1}^\Theta
\left[-\vectorize(\Bb\Db_\theta\Bb^T)^T\vectorize(\yb\yb^H)\right]\bf{1}\{\bb\in\cc_\theta\}\nonumber\\
&=&\arg\min_{\bb\in{\mathcal{C}}}\sum_{\theta=1}^\Theta
\left[-\tr\left([(\Bb\otimes\Bb)\vectorize(\Db_\theta)]^T\vectorize(\yb\yb^H)\right)\right]\bf{1}\{\bb\in\cc_\theta\}\nonumber\\
&=&\arg\min_{\bb\in{\mathcal{C}}}\sum_{\theta=1}^\Theta
\left[-\tr\left(\vectorize(\Db_\theta)^T(\Bb^T\otimes\Bb^T)\vectorize(\yb\yb^H)\right)\right]\bf{1}\{\bb\in\cc_\theta\}\nonumber\\
&=&\arg\min_{\bb\in{\mathcal{C}}}\sum_{\theta=1}^\Theta
\left[-\tr\left((\Bb\otimes\Bb)^T\vectorize(\yb\yb^H)\vectorize(\Db_\theta)^T\right)\right]\bf{1}\{\bb\in\cc_\theta\},\label{trBByyHH}
\end{eqnarray}
where ``$\otimes$'' is the Kronecker product, and ${\bf
1}\{\cdot\}$ is the set indicator function that has been used in
Lemma~\ref{Lem1}. Defining
$$\Eb\triangleq\begin{bmatrix}
0 & 0 \cdots & 0 & 0\\
1 & 0 \ddots & 0 & 0\\
0 & 1 \ddots & 0 & 0\\
0 & 0 \cdots & 1 & 0
\end{bmatrix}_{L \times L}
\quad{\rm and}\quad
\cb=\begin{bmatrix}b_1\\\vdots\\b_N\\0\\ \vdots\\
0\end{bmatrix}_{L\times 1},
$$
we get:
\begin{eqnarray*}
\lefteqn{\Bb \otimes \Bb}\\&=&\begin{bmatrix}\cb\otimes\Bb &
(\Eb\cb)\otimes \Bb & \cdots & (\Eb^{P-1}\cb)\otimes \Bb\end{bmatrix}\\
&=&\begin{bmatrix}\cb\otimes\cb & \cb\otimes(\Eb\cb) & \cdots &
\cb\otimes(\Eb^{P-1}\cb) & (\Eb\cb)\otimes \cb & \cdots &
(\Eb^{P-1}\cb)\otimes
(\Eb^{P-1}\cb)\end{bmatrix}\\
&=&\bigg[\vectorize(\cb\cb^T)\ \ \vectorize((\Eb\cb)\cb^T)
\ \cdots\ \vectorize((\Eb^{P-1}\cb)\cb^T)\ \
\vectorize(\cb(\Eb\cb)^T)\ \cdots\
\vectorize((\Eb^{P-1}\cb)(\Eb^{P-1}\cb)^T)\bigg],
\end{eqnarray*}
which indicates that the $i$th column of $\Bb\otimes\Bb$,
where $i=0,1,\cdots,P^2-1$, can be written as
$\vectorize\left((\Eb^{i\bmod
P}\cb)(\Eb^{\lfloor i/P\rfloor}\cb)^T\right).$
Here, we adopt $\Eb^0\cb=\cb$ by convention.

Resume the derivation in \eqref{trBByyHH} by denoting
the matrix entry of $\Db_\theta$ by $\delta_{i,j}^{(\theta)}$:
\begin{eqnarray*}
\hat{\bb}
&=&\arg\min_{\bb\in{\mathcal{C}}}
\sum_{\theta=1}^\Theta\left[-\sum_{i=0}^{P-1}\sum_{j=0}^{P-1}\delta_{i,j}^{(\theta)}\vectorize((\Eb^i\cb)(\Eb^j\cb)^T)^T\vectorize(\yb\yb^H)
\right]\bf{1}\{\bb\in\cc_\theta\}\\
&=&\arg\min_{\bb\in{\mathcal{C}}}
\sum_{\theta=1}^\Theta\left[-\sum_{i=0}^{P-1}\sum_{j=0}^{P-1}\delta_{i,j}^{(\theta)}\tr((\Eb^j\cb)(\Eb^i\cb)^T\yb\yb^H)\right]\bf{1}\{\bb\in\cc_\theta\}\\
&=&\arg\min_{\bb\in{\mathcal{C}}}
\sum_{\theta=1}^\Theta\left[-\sum_{i=0}^{P-1}\sum_{j=0}^{P-1}\delta_{i,j}^{(\theta)}\tr\left((\Eb^i)^T\yb\yb^H\Eb^j\cb\cb^T\right)\right]\bf{1}\{\bb\in\cc_\theta\}\\
&=&\arg\min_{\bb\in{\mathcal{C}}}
\sum_{\theta=1}^\Theta\left[-\tr(\Wb_\theta\cb\cb^T)\right]\bf{1}\{\bb\in\cc_\theta\},
\end{eqnarray*}
where
$$
\Wb_\theta\triangleq\sum_{i=0}^{P-1}\sum_{j=0}^{P-1}\delta_{i,j}^{(\theta)}(\Eb^i)^T\yb\yb^H\Eb^j.
$$
We then conclude:
\begin{eqnarray}
\hat\bb&=&\arg\min_{\bb\in{\mathcal{C}}}\sum_{\theta=1}^\Theta\left[-\vectorize(\mathbb{W}_\theta)^H\vectorize(\cb\cb^T)\right]
{\bf 1}\{\bb\in\cc_\theta\}\nonumber\\
&=&\arg\min_{\bb\in{\mathcal{C}}}
\sum_{\theta=1}^\Theta
\left[-\vectorize(\mathbb{W}_\theta)^H\vectorize(\cb\cb^T)-\vectorize(\cb\cb^T)^T\vectorize(\mathbb{W}_\theta)\right]
{\bf 1}\{\bb\in\cc_\theta\}\nonumber\\
&=&\arg\min_{\bb\in{\mathcal{C}}}\frac 12
\sum_{\theta=1}^\Theta\left[\sum_{m=1}^{N}\sum_{n=1}^{N}\left(-w_{m,n}^{(\theta)}b_mb_n\right)\right]
{\bf 1}\{\bb\in\cc_\theta\},\label{originalg}
\end{eqnarray}
where $w_{m,n}^{(\theta)}$ is the real part of the entry of $\Wb_\theta$, and is given by:
$$
w_{m,n}^{(\theta)}=\sum_{i=0}^{P-1}\sum_{j=0}^{P-1}\delta_{i,j}^{(\theta)}\mbox{Re}\{y_{m+i}y_{n+j}^\ast\}.
$$
The maximum-likelihood decision remains unchanged by adding a constant, independent of the codeword $\bb$; hence,
a constant is added to make non-negative the decision criterion as:\footnote{
Here, a
{\em non-negative} maximum-likelihood criterion makes
possible the later definition of path metric $g(\bb_{(\ell)})$
{\em non-decreasing} along any path in the code tree. A
non-decreasing path metric has been shown to be a sufficient
condition for priority-first search to guarantee
to locate the codeword with the smallest path metric
\cite{Han}\cite{HCW}. It can then be anticipated
(cf.~Section~\ref{SEC-IV-B}) that letting the heuristic function
be zero for all paths in the code tree suffices to result in an
evaluation function satisfying the optimal condition
\eqref{optimalcon} in Lemma~\ref{lemma0}.

Notably, the additive constant that makes the evaluation function
non-decreasing along any path in the code tree can also be
obtained by first defining $g$ based on \eqref{originalg}, and
then determining its respective $h$ according to
\eqref{optimalcon}. Such an approach however makes complicate
the determination of heuristic function $h$ when the system
constraint that the evaluation function is recursive-computable is
additionally required. The alternative approach that directly
defines a recursive-computable $g$ based on a non-negative
maximum-likelihood criterion is accordingly adopted in this work.
}
\begin{eqnarray*}
\hat\bb&=&
\arg\min_{\bb\in{\mathcal{C}}}\left\{
\sum_{m=1}^N\max_{1\leq\eta\leq\Theta}\left(\sum_{n=1}^{m-1}|w_{m,n}^{(\eta)}|+\frac 12|w_{m,m}^{(\eta)}|\right)-
\frac 12\sum_{\theta=1}^\Theta\left[\sum_{m=1}^{N}\sum_{n=1}^{N}w_{m,n}^{(\theta)}b_mb_n\right]
{\bf 1}\{\bb\in\cc_\theta\}\right\}\\
&=&
\arg\min_{\bb\in{\mathcal{C}}}\sum_{\theta=1}^\Theta\left[
\sum_{m=1}^{N}\max_{1\leq\eta\leq\Theta}\left(\sum_{n=1}^{m-1}|w^{(\eta)}_{m,n}|+\frac 12|w_{m,m}^{(\eta)}|\right)-
\frac 12\sum_{m=1}^{N}\sum_{n=1}^{N}w_{m,n}^{(\theta)}b_mb_n
\right]{\bf 1}\{\bb\in\cc_\theta\}.
\end{eqnarray*}
It remains to prove that the metric of
$$\sum_{m=1}^{N}\max_{1\leq\eta\leq\Theta}\left(\sum_{n=1}^{m-1}|w^{(\eta)}_{m,n}|+\frac 12|w_{m,m}^{(\eta)}|\right)-
\frac 12\sum_{m=1}^{N}\sum_{n=1}^{N}w_{m,n}^{(\theta)}b_mb_n$$
can be computed recursively
for $\bb\in\cc_\theta$.

Define for path $\bb_{(\ell)}$ over code tree $\theta$ that
\begin{equation}
\label{metricg}
g(\bb_{(\ell)})\triangleq\sum_{m=1}^{\ell}\max_{1\leq\eta\leq\Theta}\left(\sum_{n=1}^{m-1}|w^{(\eta)}_{m,n}|+\frac
12|w_{m,m}^{(\eta)}|\right)-
\frac 12\sum_{m=1}^{\ell}\sum_{n=1}^{\ell}w_{m,n}^{(\theta)}b_mb_n.
\end{equation}
Then, by the symmetry that $w_{m,n}^{(\theta)}=w_{n,m}^{(\theta)}$
for $1\leq m,n\leq N$ and $1\leq\theta\leq\Theta$,
we have that for $1\leq\ell\leq N-1$,
\begin{eqnarray}
g(\bb_{(\ell+1)})&=&g(\bb_{(\ell)})+ \max_{1\leq\eta\leq\Theta}
\left(\sum_{n=1}^{\ell}|w^{(\eta)}_{\ell+1,n}|
+\frac 12|w^{(\eta)}_{\ell+1,\ell+1}|\right)-\sum_{n=1}^\ell w_{\ell+1,n}^{(\theta)}b_{\ell+1}b_n-\frac 12 w_{\ell+1,\ell+1}^{(\theta)}\nonumber\\
&=&g(\bb_{(\ell)})+\max_{1\leq\eta\leq\Theta}\alpha_{\ell+1}^{(\eta)}
-b_{\ell+1}\sum_{i=0}^{P-1}\sum_{j=0}^{P-1}\delta_{i,j}^{(\theta)}\mathrm{Re}
\left\{y_{\ell+i+1}\cdot
u_j(\bb_{(\ell+1)})\right\},\label{g}
\end{eqnarray}
where
\begin{eqnarray}
\alpha_{\ell+1}^{(\eta)}&\triangleq&\sum_{n=1}^{\ell}|w^{(\eta)}_{\ell+1,n}|
+\frac 12|w^{(\eta)}_{\ell+1,\ell+1}|\nonumber\\
&=&\sum_{n=1}^{\ell}
\left|\sum_{i=0}^{P-1}\sum_{j=0}^{P-1}\delta_{i,j}^{(\eta)}\mbox{Re}\{y_{\ell+i+1}y_{n+j}^\ast\}\right|
+\frac
12\left|\sum_{i=0}^{P-1}\sum_{j=0}^{P-1}\delta_{i,j}^{(\eta)}\mbox{Re}\{y_{\ell+i+1}y_{\ell+j+1}^\ast\}\right|
\label{alpha1}
\end{eqnarray}
and for $0\leq j\leq P-1$,
$$
u_j(\bb_{(\ell+1)})\triangleq\sum_{n=1}^{\ell}b_n
y_{n+j}^{\ast}+\frac 12b_{\ell+1}y_{\ell+j+1}^\ast=u_j(\bb_{(\ell)})
+\frac 12\left(b_\ell y_{\ell+j}^\ast+b_{\ell+1}y_{\ell+1+j}^\ast\right).
$$
This implies that we can
recursively compute $g(\bb_{(\ell+1)})$
and $\{u_j(\bb_{(\ell+1)})\}_{0\leq j\leq P-1}$ from the previous
$g(\bb_{(\ell)})$
 and
$\{u_j(\bb_{(\ell)})\}_{j=0}^{P-1}$ with the knowledge of
 $y_{\ell+1}$, $y_{\ell+2}$, $\cdots$, $y_{\ell+P}$ and
$b_{\ell+1}$, and the initial condition satisfies that
$g(\bb_{(0)})=u_j(\bb_{(0)})=b_0=0$ for
$0\leq j\leq P-1$.

A final remark in this discussion is that
although
the computation burden of $\alpha_\ell^{(\eta)}$ in \eqref{alpha1} increases linearly with $\ell$,
such a linearly growing load can be moderately compensated by the fact that
$\alpha_\ell^{(\eta)}$ is only necessary to compute it once for each $\ell$ and $\eta$, because
it can be shared for
all paths ending at level $\ell$ over code tree $\eta$.

\subsection{Heuristic function $h$ that validates \eqref{optimalcon}}
\label{SEC-IV-B}

Taking the maximum-likelihood metric $g$ into the sufficient condition in \eqref{optimalcon} yields that:
\begin{eqnarray*}
\lefteqn{\sum_{m=1}^{\ell}\max_{1\leq\eta\leq\Theta}\alpha_m^{(\eta)}-
\frac 12\sum_{m=1}^{\ell}\sum_{n=1}^{\ell}w_{m,n}^{(\theta)}b_mb_n
+h(\bb_{(\ell)})}\\
&\leq&
\min_{\left\{\tilde\bb\in\cc~:~\tilde\bb_{(\ell)}=\bb_{(\ell)}\right\}}
\left[\sum_{m=1}^{N}\max_{1\leq\eta\leq\Theta}\alpha_m^{(\eta)}-
\frac 12\sum_{m=1}^{N}\sum_{n=1}^{N}w_{m,n}^{(\theta)}b_mb_n+h(\tilde\bb)\right].
\end{eqnarray*}
Hence, in addition to $h(\tilde\bb)=0$, the heuristic function should satisfy:
\begin{equation}
h(\bb_{(\ell)})
\leq\sum_{m=\ell+1}^{N}\max_{1\leq\eta\leq\Theta}\alpha_m^{(\eta)}-\max_{\left\{\tilde\bb\in\cc:\tilde\bb_{(\ell)}=\bb_{(\ell)}\right\}}
\left(\sum_{m=\ell+1}^N\tilde b_m\sum_{n=1}^\ell w_{m,n}^{(\theta)} b_n
+\frac 12\sum_{m=\ell+1}^N\sum_{n=\ell+1}^N w_{m,n}^{(\theta)}\tilde b_m\tilde b_n\right).\label{target}
\end{equation}

Apparently, a function that guarantees to satisfy \eqref{target} is the zero-heuristic function, that is,
$h_1(\bb_{(\ell)})=0$ for any path $\bb_{(\ell)}$ in the code trees.
Adopting the zero-heuristic function $h_1$, together with the recursively computable
maximum-likelihood metric $g$ in \eqref{metricg}, makes feasible the {\em on-the-fly} priority-first search decoding.
In comparison with the exhaustive-checking decoding,
significant improvement in the computational complexity is resulted especially at medium-to-high SNRs.

In situation when the codeword length $N$ is not large such as $N\leq 50$
so that the demand of {\em on-the-fly} decoding can be moderately relaxed, we can adopt
a larger heuristic function to further reduce the computational complexity.
Upon the reception
of all $y_1,\cdots,y_L$, the heuristic function that satisfies
\eqref{target} regardless of $\tilde b_{\ell+1}$,
$\cdots$, $\tilde b_N$ can be increased up to:
\begin{eqnarray}
h_2(\bb_{(\ell)})&\triangleq&\sum_{m=\ell+1}^{N}\max_{1\leq\eta\leq\Theta}\alpha_m^{(\eta)}
-\sum_{m=\ell+1}^{N}\left|\sum_{n=1}^{\ell}w_{m,n}^{(\theta)}b_n\right|
-\frac 12\sum_{m=\ell+1}^N\sum_{n=\ell+1}^{N}\left|w_{m,n}^{(\theta)}\right|\nonumber\\
&=&\sum_{m=\ell+1}^{N}\max_{1\leq\eta\leq\Theta}\alpha_m^{(\eta)}-\sum_{m=\ell+1}^{N}\left|v_m^{(\theta)}(\bb_{(\ell)})\right|-\beta_\ell^{(\theta)},\label{heu2}
\end{eqnarray}
where for $1\leq\ell,m\leq N$ and $1\leq\theta\leq\Theta$,
$$
v_m^{(\theta)}(\bb_{(\ell)})\triangleq\sum_{n=1}^{\ell}w_{m,n}^{(\theta)}b_n
=v_{m}^{(\theta)}(\bb_{(\ell-1)})+b_{\ell}w_{\ell,m}^{(\theta)}
$$
and
$$
\beta_\ell^{(\theta)}\triangleq\sum_{m=\ell+1}^{N}
\left(\sum_{n=\ell+1}^{m-1}|w_{m,n}^{(\theta)}|+\frac 12|w_{m,m}^{(\theta)}|\right)
=\beta_{\ell-1}^{(\theta)}-\sum_{n=\ell+1}^N|w_{\ell,n}^{(\theta)}|-\frac 12|w_{\ell,\ell}^{(\theta)}|
$$
with initially $v_m^{(\theta)}(\bb_{(0)})=b_0=0$, and
$\beta_0^{(\theta)}=\sum_{m=1}^{N}\alpha_m^{(\theta)}$.
Simulations show that when being compared with the zero-heuristic
function $h_1$, the heuristic function in \eqref{heu2} further reduces
the number of path expansions during the decoding process up to
one order of magnitude (cf.~Tab.~\ref{table3a}, in which
$f_1=g+h_1=g$ and $f_2=g+h_2$.).

A final note on the priority-first search of the maximum-likelihood codeword is
that in those cases that equality in \eqref{coderule} cannot be fulfilled,
codewords will be selected equally from multiple code trees, e.g.,
one code tree structured according to $\Bb^T\Bb=\Gb_1$, and the other code tree
targeting $\Bb^T\Bb=\Gb_2$ for $N$ even and $P=2$.
Since the transmitted codeword belongs to only one of the code trees,
to maintain {\em individual Stack} for the codeword search over {\em each} code tree
will introduce considerable unnecessary decoding burdens
especially for the code trees that the transmitted codeword does not belong to.
Hence, only one Stack is maintained during the priority-first search,
and the evaluation function values for different code trees are
compared and sorted in the same Stack.
The path to be expanded next is therefore the one whose evaluation function value
is globally the smallest.

\section{Simulation Results}
\label{Simulationresults}

In this section, the performance of the rule-based constructed
codes proposed in Section \ref{Codeconstruction} is examined. Also
illustrated is the decoding complexity of the maximum-likelihood
priority-first search decoding algorithm presented in the previous section.
For ease of comparison, the channel parameters used in our
simulations follow those in \cite{Skoglund}, where $\hb$ is
complex zero-mean Gaussian distributed with
$E[\hb\hb^H]=(1/P)\Ib_P$ and $P=2$. The average system SNR is thus
given by:
\begin{equation}
\mbox{Average SNR}=\frac
N{L}\frac 1{\sigma_n^2}\tr\left(E[\hb\hb^H]\frac{1}{N}\Bb^T\Bb\right) =\frac
N{L}\frac 1{\sigma_n^2}\tr\left(\frac{1}{NP}\Bb^T\Bb\right)=\frac
N{(N+P-1)}\frac 1{\sigma_n^2}, \label{AveSNR}
\end{equation}
since $\tr\left(\Bb^T\Bb\right)=NP$ for all codewords simulated.\footnote{
The authors in \cite{Skoglund} directly define the channel SNR as $1/\sigma_n^2$.
It is apparent that their definition is exactly the limit of \eqref{AveSNR} as $N$ approaches infinity.

Since it is assumed that adequate guard period between two encoding blocks exists
(so that there is no interference between two consecutive decoding blocks),
the computation of the system SNR for finite $N$ should be adjusted to account for this
muting (but still part-of-the-decoding-block) guard period.
For example, in comparison of the (6,3) and (20,10) codes over channels with memory order $1$ (i.e., $P=2$),
one can easily observe that the former can only transmit 18 code bits in the time interval of $21$ code bits,
while the latter pushes out up to 20 code bits in the period of the same duration.
Thus, under fixed code bit transmission power and fixed component noise power $\sigma_n^2$,
it is reasonable for the (20,10) code to result a higher SNR than the (6,3) code.
}

\begin{figure}[tbp]
\centering \includegraphics[width=5in]{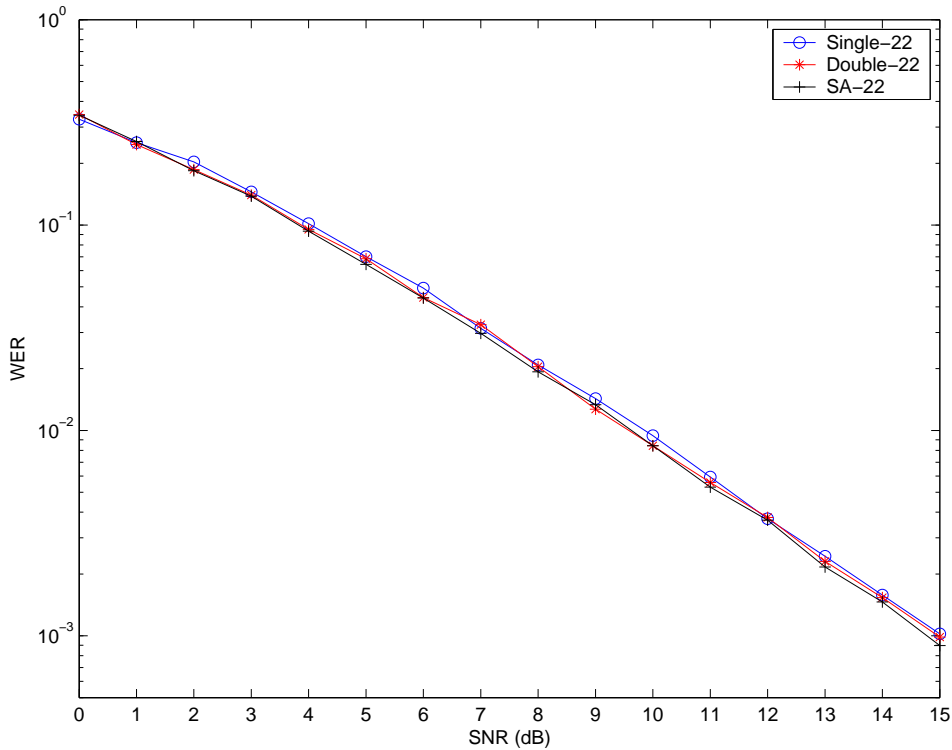}
\caption{The maximum-likelihood word error rates (WERs) of the half-rate computer-searched code
by simulated annealing in \cite{Skoglund} (SA-22),
the rule-based half-rate code with double code trees (Double-22),
and the rule-based half-rate code with single code tree (Single-22).
The codeword length is $N=22$.} \label{fig:WER}
\end{figure}

\begin{figure}[tbp]
\centering \includegraphics[width=5in]{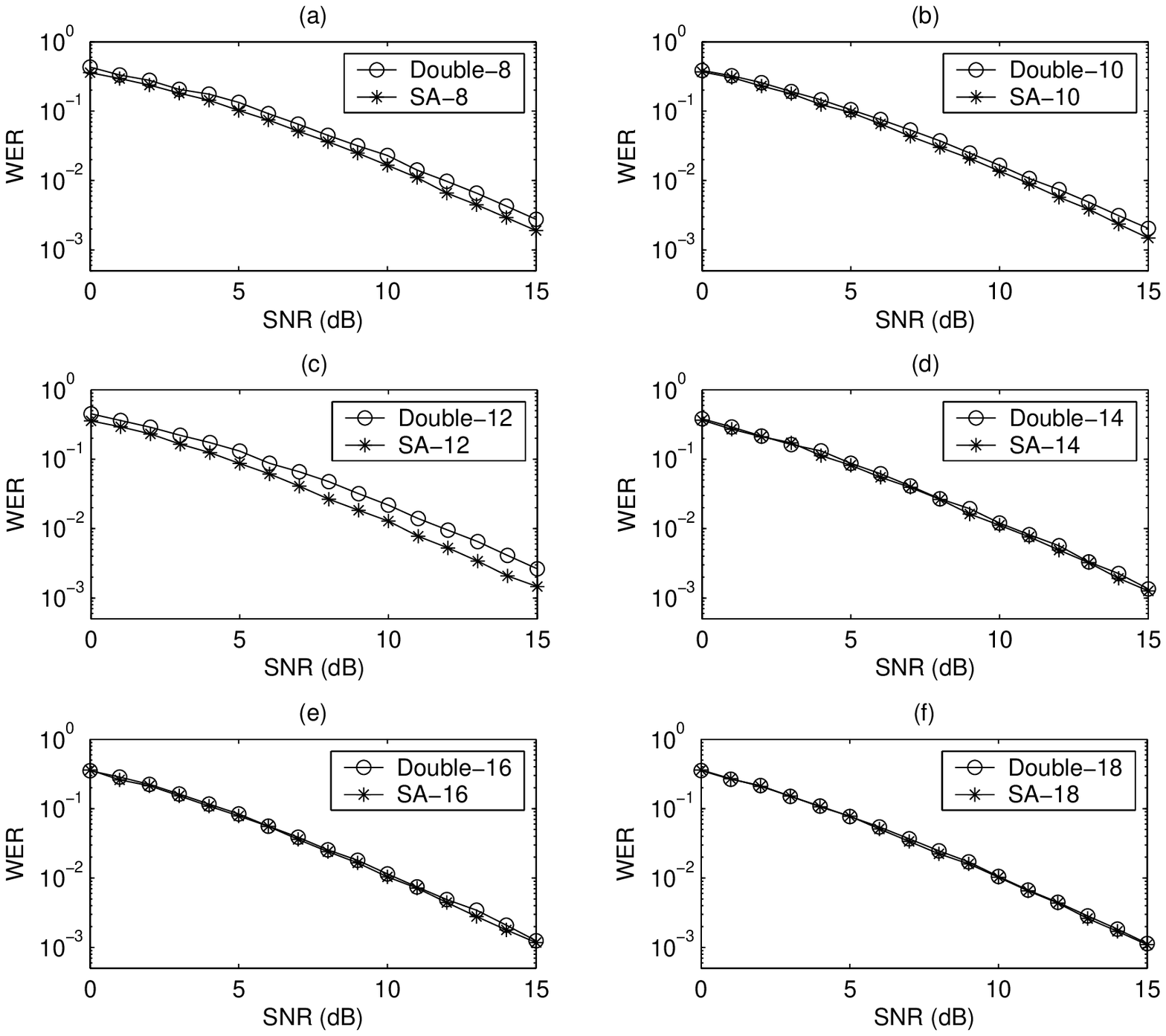} \caption{The
maximum-likelihood word error rates (WERs) of the
computer-searched half-rate code by simulated annealing (SA-$N$) and the
rule-based half-rate code with double code trees (Double-$N$).}
\label{fig:WER1}
\end{figure}

There are three codes simulated in Fig.~\ref{fig:WER}: the computer-searched half-rate code obtained
in~\cite{Skoglund} (SA-22), the rule-based double-tree code in which half of the
codewords satisfying $\Bb^T\Bb=\Gb_1$ and the remaining half
satisfying $\Bb^T\Bb=\Gb_2$ (Double-22), and  the rule-based single-tree code whose codewords
are all selected from the candidate sequences satisfying
$\Bb^T\Bb=\Gb_1$ (Single-22).
We observe from Fig.~\ref{fig:WER} that the Double-22 code performs almost the same as the SA-22 code obtained in \cite{Skoglund} at $N=22$. Actually, extensive
simulations in Fig.~\ref{fig:WER1} show that the performance of
the rule-based double-tree half-rate codes is as good as the computer-searched
half-rate codes for all $N>12$. However, when $N\leq 12$, the
approximation in \eqref{ratio1} can no longer be well maintained due to
the restriction that $|\ac(\bb_{(\ell)}|\Gb)|$ must be an integer,
and an apparent performance deviation between the rule-based double-tree
half-rate codes and the computer-searched half-rate codes can therefore be
sensed for $N$ below 12.

In addition to the Double-22 code, the performance of the Single-22 code is also simulated in Fig.~\ref{fig:WER}.
Since the pairwise codeword distance
in the sense of \eqref{pwdis} for the Single-22 code is in general
smaller than that of the Double-22 code, its performance has $0.2$ dB degradation to that of the Double-22 code. However, we will see in later simulation that the
Single-22 code has the smallest decoding complexity among
the three codes in Fig.~\ref{fig:WER}.
This suggests that to select codewords uniformly from a single
code tree should not be ruled out as a candidate design, especially
when the decoding complexity becomes
the main system concern.

\begin{figure}[tbp]
\centering \includegraphics[width=5in]{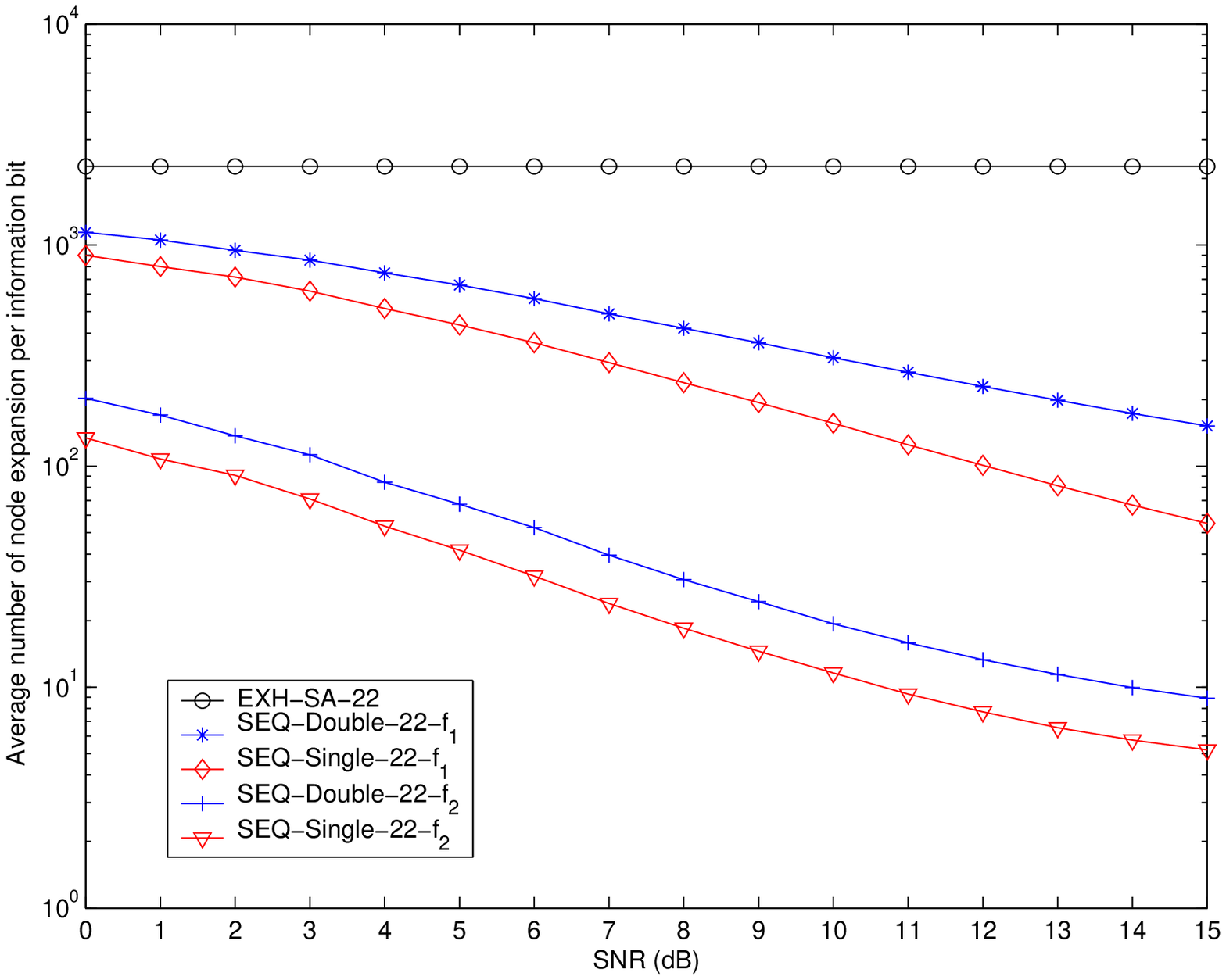} \caption{The
average numbers of node expansions per information bit
for the computer-searched code in \cite{Skoglund} by exhaustive
decoding (EXH-SA-22), and the rule-based single-tree (SEQ-Single-$22$) and double-tree (SEQ-Double-$22$) codes using the
priority-first search decoding guided by either evaluation function $f_1$ or evaluation function $f_2$.} \label{fig:complexity}
\end{figure}

In Fig.~\ref{fig:complexity}, the average numbers of node expansions
per information bit are illustrated for the codes examined in Fig.~\ref{fig:WER}.
Since the number of nodes expanded is exactly the number of
tree branch metrics (i.e., one recursion of $f$-function values) computed,
the equivalent complexity of exhaustive decoder is correspondingly plotted.
 It can then be observed that in comparison with the exhaustive decoder,
a significant reduction in computational burdens can be
obtained at moderate-to-high SNRs by adopting
the Double-22 code and
the priority-first search decoder with on-the-fly evaluation function $f_1$, namely, $g$ (cf.~Eq.~\eqref{g}).
Further reduction can be approached if the Double-22 code is replaced with the Single-22 code.
The is because performing the sequential search over multiple code trees
introduce extra node expansions for those code trees
that the transmitted codeword does not belong to.
An additional order-of-magnitude reduction in node expansions can be achieved
when the evaluation function $f_2=g+h_2$ is used instead.

\begin{figure}[tbp]
\centering \includegraphics[width=5in]{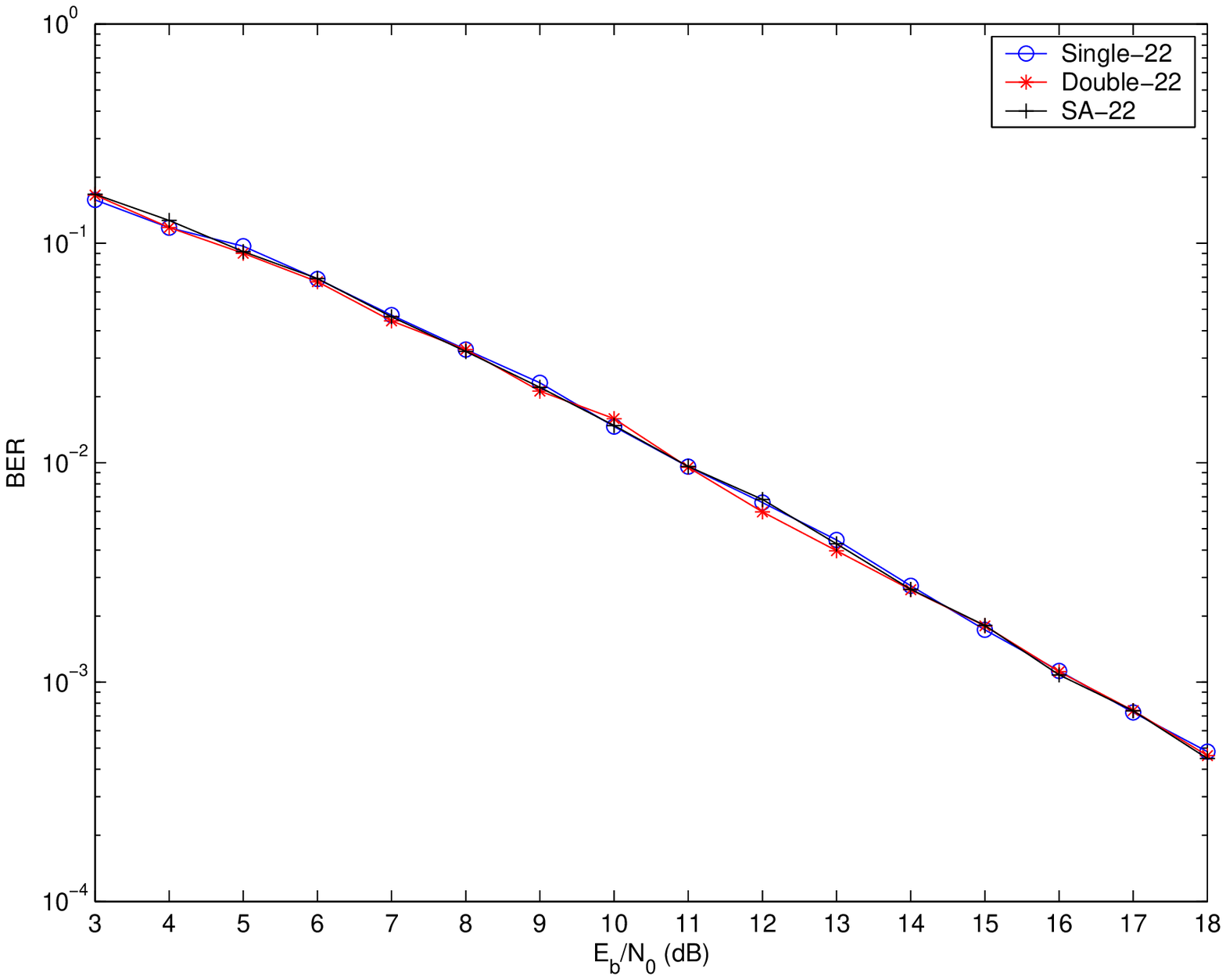} \caption{Bit
error rates (BERs) for the codes simulated in Fig.~\ref{fig:WER}.}
\label{fig:fig4}
\end{figure}

The authors in \cite{Coskun} and \cite{Skoglund} only focused on
the word-error-rates (WERs). No bit error rate (BER)
performances that involve the mapping design between
the information bit patterns and the codewords
were presented. Yet, in certain applications, such as voice transmission
and digital radio broadcasting, the BER is generally considered a more critical performance index.
In addition, the adoption of the BER performance index, as well as the signal-to-noise ratio per information bit,
facilitates the comparison among codes of differen code rates.

Figure \ref{fig:fig4} depicts the BER performances of the codes simulated
in Fig.~\ref{fig:WER}. The corresponding $E_b/N_0$ is computed according to:
$$
E_b/N_0=\frac 1{R}\cdot\mbox{SNR},
$$
where $R=K/N$ is the code rate.
The mapping between the bit patterns and the codewords of the given computer-searched code is obtained
through simulated annealing by minimizing the upper bound of:
$$
\mbox{BER}\leq\frac 1{2^{K}}\sum_{i=1}^{2^K}\sum_{j=1,j\neq i}^{2^K}\frac{D(\mb(i),\mb(j))}{K}
\Pr\left(\left.\hat{\bb}=\bb(j)\right|\bb(i)\mbox{ transmitted}\right),
$$
where, other than the notations defined in \eqref{pwp},
 $\mb(i)$ is the information sequence corresponding to $i$-th codeword, and
$D(\cdot,\cdot)$ is the Hamming distance.
For the rule-based constructed codes in Section \ref{Encode-Subsection},
the binary representation of the index of the requested codeword in Step \ref{Step1}
is directly taken as the information bit pattern corresponding to the requested codeword.
The result in Fig.~\ref{fig:fig4} then indicates that
the BER performances of the three curves are almost the same, which directs
the conclusion
that taking the binary representation of the requested codeword index as the
information bit pattern for the rule-based constructed code not only makes
easy its implementation but also has similar BER performance to the computer-optimized codes.

\begin{figure}[tbp]
\centering \includegraphics[width=5in]{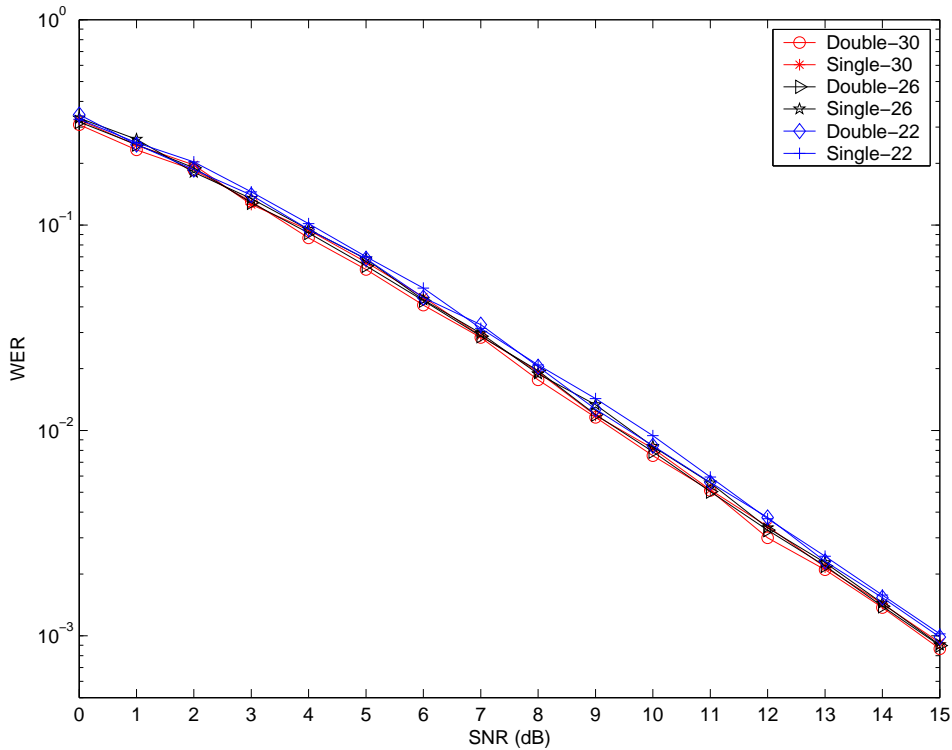} \caption{Word
error rates (WERs) for the codes of Single-$22$, Double-$22$,
Single-$26$, Double-$26$, Single-$30$ and Double-$30$.}
\label{fig:sim5}
\end{figure}

\begin{figure}[tbp]
\centering \includegraphics[width=5in]{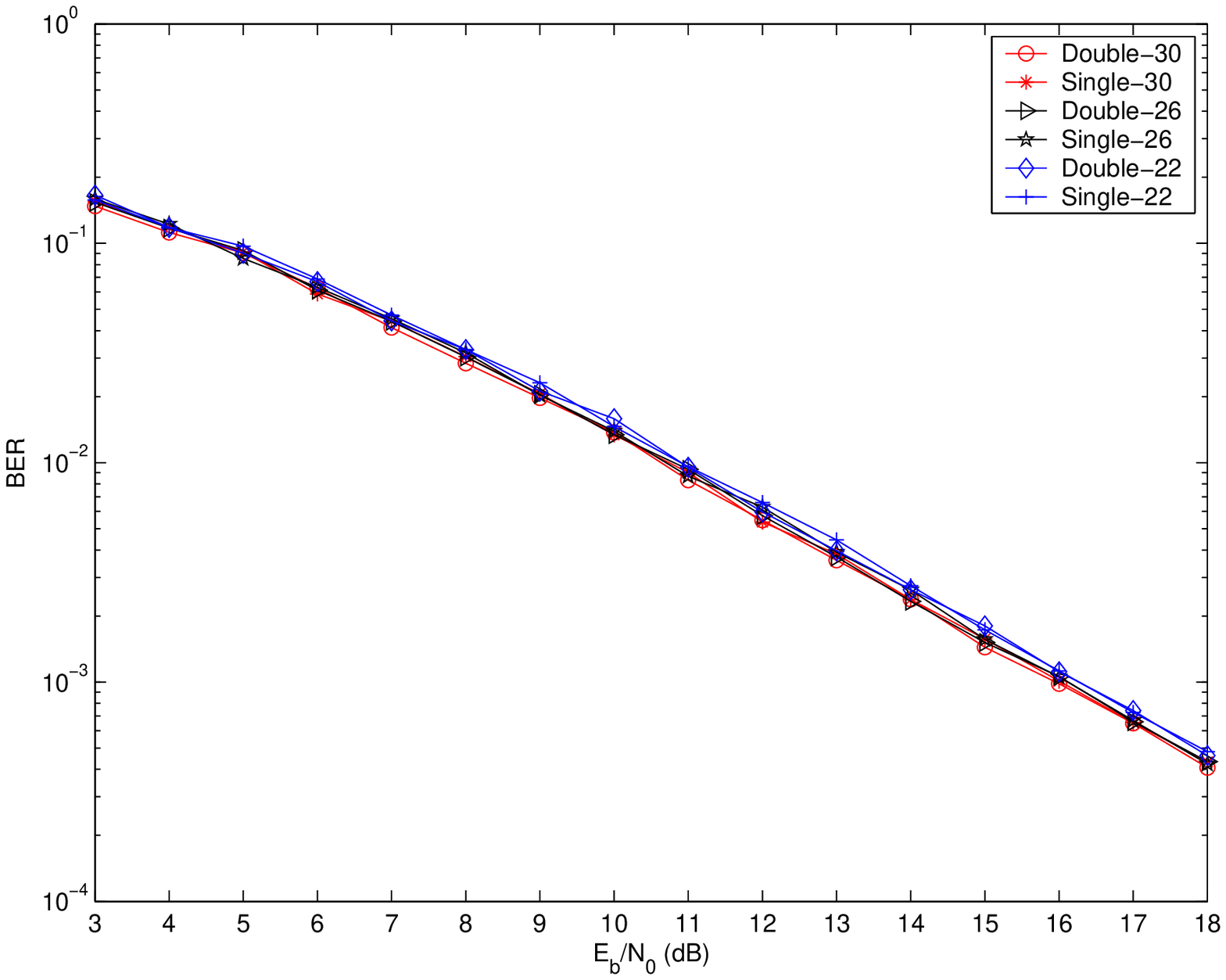} \caption{Bit
error rates (BERs) for the codes of Single-$22$, Double-$22$,
Single-$26$, Double-$26$, Single-$30$ and Double-$30$.}
\label{fig:sim6}
\end{figure}

In the end,
we demonstrate the
WER and BER performances of Single-$26$, Double-$26$, Single-$30$,
Double-$30$ codes, together with those of Single-$22$ and Double-$22$ codes,
over the quasi-static fading channels respectively in Figs.~\ref{fig:sim5} and \ref{fig:sim6}.
Both figures show that the Double-$30$ code
has the best maximum-likelihood performance not only in WER but also in BER.
This result echoes the usual anticipation that
the performance favors a longer code
as long as the channel coefficients
remain unchanged in a coding block.
Their decoding complexities are listed in Tab.~\ref{table3a},
from which we observe that the saving of decoding complexity
of metric $f_2$ with respect to metric $f_1$ increases
as the codeword length further grows.

\begin{table}[hbt]
\caption{Average numbers of
node expansions per information bit for the priority-first search decoding of
the constructed half-rate codes of length $22$, $26$ and $30$.}\label{table3a}
\begin{center}
\begin{tabular}{||c|r|r|r|r|r|r|r|r|r|r|r||}\hline\hline
$SNR$&5dB&6dB&7dB&8dB&9dB&10dB&11dB&12dB&13dB&14dB&15dB\\\hline
Double-22-$f_1$&671&590&506&436&375&320&274&236&204&178&156\\[-1.5mm]
Double-22-$f_2$&68&55&42&32&26&20&17&14&12&10&9\\\hline ratio of
$\left.f_1\right/f_2$
&9.8&10.7&12.0&13.6&14.4&16.0&16.1&16.8&17.0&17.8&17.3\\\hline\hline
Double-$26$-$f_1$&$2361$&$2006$&$1695$&$1416$&$1189$&$981$&$813$&$677$&523&499&392\\[-1.5mm]
Double-$26$-$f_2$&$175$&$130$&$94$&$69$&$53$&$39$&$29$&$23$&18&15&13\\\hline
ratio of $\left.f_1\right/f_2$
&$13.5$&$15.4$&$18.0$&$20.5$&$22.4$&$25.2$&$28.0$&$29.4$&29.1&33.3&30.2\\\hline\hline
Double-$30$-$f_1$&$8455$&$7073$&$5760$&$5133$&$3759$&$3430$&$2644$&$1996$&$1765$&$1368$&$1081$\\[-1.5mm]
Double-$30$-$f_2$&$459$&$332$&$232$&$166$&$119$&$86$&$60$&$44$&$33$&$25$&$20$\\\hline
ratio of $\left.f_1\right/f_2$
&$18.4$&$21.3$&$24.8$&$30.9$&$31.6$&$39.9$&$44.1$&$45.4$&$53.4$&$54.7$&$54.1$\\\hline\hline
Single-$22$-$f_1$&460&371&308&250&200&163&130&105&85&69&57\\[-1.5mm]
Single-$22$-$f_2$&45&33&26&20&15&12&10&8&7&6&5\\\hline ratio
of$\left.f_1\right/f_2$&10.2&11.2&11.8&12.5&13.3&13.5&13.0&13.1&12.1&11.5&11.4\\\hline\hline
Single-$26$-$f_1$&$1635$&$1328$&$1061$&$839$&$666$&$522$&$403$&$312$&$244$&$191$&$152$\\[-1.5mm]
Single-$26$-$f_2$&$112$&$79$&$57$&$42$&$31$&$23$&$17$&$13$&$11$&$9$&$7$\\\hline
ratio of $\left.f_1\right/f_2$
&$14.6$&$16.8$&$18.6$&$20.0$&$21.5$&$22.7$&$23.7$&$23.9$&$22.2$&$21.2$&$21.7$\\\hline\hline
Single-$30$-$f_1$&$5871$&$4695$&$3857$&$2924$&$2335$&$1813$&$1328$&$884$&$805$&$572$&$416$\\[-1.5mm]
Single-$30$-$f_2$&$284$&$199$&$144$&$101$&$72$&$51$&$35$&$26$&18&14&11\\\hline
ratio of $\left.f_1\right/f_2$
&$20.6$&$23.6$&$26.8$&$29.0$&$32.4$&$35.5$&$38.0$&$34.0$&$44.7$&$40.9$&$37.8$\\\hline\hline
\end{tabular}
\end{center}
\end{table}

\section{Codes for channels with fast fading}
\label{fastfading}

In previous sections, also in \cite{Chugg}, \cite{Giese} and
\cite{Skoglund}, it is assumed that the channel coefficients $\hb$
are invariant in each coding block of length $L=N+P-1$. In this
section, we will show that the approaches employed in previous
sections can also be applicable to the situation that $\hb$ may
change in every $Q$ symbol, where $Q<L$.

For $1\leq k\leq M=\lceil L/Q\rceil$,
let $\hb_k\triangleq[h_{1,k}\ \ h_{2,k}\ \cdots\ h_{P,k}]^T$ be the constant channel coefficients
at the $k$th sub-block. Denote by
$\bb_k=[b_{(k-1)Q-P+2}\ \cdots\ b_{(k-1)Q+1}\ \cdots\ b_{kQ}]^T$
the portion of $\bb$, which will affect the output portion
$\yb_k=[y_{(k-1)Q+1}\ \ y_{(k-1)Q+2}\ \cdots\ y_{kQ}]$,
where we assume $b_j=0$ for $j\leq 0$ and $j>N$ for notational convenience.
Then, for a channel whose coefficients change in every $Q$ symbol,
the system model defined in \eqref{systemmodel}
remains as $\yb=\Bb\hb+\nb$ except that both $\yb$ and $\nb$ extend as $MQ\times 1$ vectors
with $y_j=n_j=0$ for $j>L$, and $\Bb$ and $\hb$ have to be re-defined as
$$\Bb\triangleq\Bb_1\oplus\Bb_2\oplus\cdots\oplus\Bb_M\quad{\rm and}\quad
\hb\triangleq\begin{bmatrix}\hb_1^H&\hb_2^H&\cdots&\hb_M^H\end{bmatrix}^H,$$
where
$\Bb_k=[\zerob_{Q\times(P-1)}\ \ \Ib_Q][\bb_k\ \ \tilde\Eb\bb_k\ \cdots\ \tilde\Eb^{P-1}\bb_k]$
is a $Q\times P$ matrix,
$$\tilde\Eb\triangleq\begin{bmatrix}
0 & 0 \cdots & 0 & 0\\
1 & 0 \ddots & 0 & 0\\
0 & 1 \ddots & 0 & 0\\
0 & 0 \cdots & 1 & 0
\end{bmatrix}_{(Q+P-1)\times(Q+P-1)}
$$
and ``$\oplus$'' is the direct sum operator of two matrices.\footnote{For two matrices $\Ab$
and $\Bb$, the direct sum of $\Ab$ and
$\Bb$ is defined as $\Ab\oplus\Bb=\begin{bmatrix}\Ab & \zerob\\
\zerob & \Bb\end{bmatrix}$.}

Based on the new system model, we have
$\Pb_B=\Pb_{B_1}\oplus \Pb_{B_2}\oplus\cdots\oplus \Pb_{B_M}$, where $\Pb_{B_k}=\Bb_k(\Bb_k^T \Bb_k)^{-1}\Bb_k^T$,
and
Eq.~\eqref{yy_Pb} becomes:
\begin{equation}
\hat{\bb}
=\arg\max_{\bb \in
\mathcal{C}}\sum_{k=1}^{M}\left\|\vectorize(\yb_k\yb_k^H)-\vectorize(\Pb_{B_k})\right\|^2\label{fastmetric}
\end{equation}
Again, codeword $\bb$ and transformed codeword $\Pb_B$ is not one-to-one corresponding
unless the first element of $\bb$, namely $b_1$, is fixed.\footnote{It can be derived that
given $Q\geq P$ and $\Bb_k^T\Bb_k=\Gb_k$ for $1\leq k\leq M$,
$$
\left\{
\begin{array}{rcl}
b_{Q-P+2}&=&b_1\times (-1)^{(Q-P+1-\gamma_{P,P-1,1})/2}\\
b_{kQ-P+2}&=&b_{(k-1)Q-P+2}\times(-1)^{(Q-\gamma_{P,P-1,k})/2}\mbox{ for }k=2,\cdots,M-1
\end{array}\right.
$$
where $\gamma_{i,j,k}$ is the $(i,j)$th entry of the symmetric matrix $\Gb_k$ for $1\leq i,j\leq P$,
and, in our setting,
$\gamma_{P,P-1,k}\in\{0,\pm 1\}$ should be chosen to make the exponent of $(-1)$ an integer.
Therefore, the first bit in each $\bb_k$ is fixed once $b_1$ is set,
which indicates that with the knowledge of $b_1$, codeword $\bb$ can be uniquely determined by
transformed codeword $\Pb_B$.
}

Since $\Bb^T\Bb=(\Bb_1^T\Bb_1)\oplus(\Bb_2^T\Bb_2)\oplus\cdots\oplus(\Bb_M^T\Bb_M)$,
the maximization of system SNR can be achieved simply by assigning
\begin{equation}
\label{MQ}
\Bb_1^T\Bb_1=\Bb_2^T\Bb_2=\cdots=\Bb_M^T\Bb_M=Q\cdot\Ib_Q
\end{equation}
if such assignment is possible. Due to the same reason mentioned in Section \ref{IIIA}, approximation to \eqref{MQ}
will have to be taken in the true code design.

It remains to determine the number of
all possible $\pm 1$-sequences of length $N$, whose first $\ell$ bits equal $b_1$,
$b_2$, $\ldots$, $b_\ell$ subject to $\Bb_k^T\Bb_k=\Gb_k$ for $1\leq k\leq M$.

\bigskip

\begin{lemma}\label{Lem5}
Fix $P=2$ and $Q\geq P$, and put
\begin{eqnarray}
\Bb_1^T\Bb_1=\begin{bmatrix}Q&c_1\\c_1&Q-1\end{bmatrix},\quad
\Bb_k^T\Bb_k=\begin{bmatrix}Q&c_k\\c_k&Q\end{bmatrix}\mbox{ for }2\leq k\leq M-1,\quad\nonumber\\
\mbox{and }\Bb_M^T\Bb_M=
\begin{bmatrix}N-(M-1)Q&c_M\\c_M&N-(M-1)Q+1\end{bmatrix},\label{aa}
\end{eqnarray}
where in our code selection process, $[c_1,c_2,\cdots,c_M]\in\{0,\pm 1\}^M$
will be chosen such that $Q-1+c_1$, $Q+c_k$
for $2\leq k\leq M-1$, and $N-(M-1)Q+c_M$ are all even.
Then, the number of
all possible $\pm 1$-sequences of length $N$, whose first $\ell$ bits equal $b_1$,
$b_2$, $\ldots$, $b_\ell$ subject to \eqref{aa} is given by:
\begin{eqnarray*}
\left\{\begin{array}{r}
\displaystyle\binom{Q-(\ell\bmod Q)}{\frac{Q-(\ell\bmod Q)+c_\tau-m_\ell}2}
\left[\displaystyle\prod_{k=\tau+1}^{M-1}\binom{Q}{\frac {Q+c_{k+1}}2}\right]\displaystyle\binom{N-(M-1)Q}{\frac{N-(M-1)Q+c_M}2}\oneb\left\{|c_\tau-m_\ell|\leq Q-(\ell\bmod Q)\right\},\\
\mbox{ for }1\leq\tau<M;\\
\displaystyle\binom{N-(M-1)Q}{\frac{N-(M-1)Q+c_M-m_\ell}2}\oneb\left\{|c_M-m_\ell|\leq N-(M-1)Q\right\},\hspace*{45mm}
\mbox{ for }\tau=M\\
\end{array}\right.
\end{eqnarray*}
where $\tau=\lfloor\ell/Q\rfloor+1$, and
$$m_\ell=\left\{\begin{array}{ll}
0,&\ell=1\mbox{ or }(\ell=(\tau-1)Q\mbox{ and }2\leq\tau\leq M);\\
b_1b_2+\cdots+b_{\ell-1}b_{\ell},&1<\ell<Q;\\
b_{(\tau-1)Q}b_{(\tau-1)Q+1}+\cdots+b_{\ell-1}b_\ell,&(\tau-1)Q<\ell< \tau Q\mbox{ and }2\leq\tau\leq M.
\end{array}\right.$$
\end{lemma}
\medskip
\begin{proof} It requires
\begin{eqnarray*}
\left\{\begin{array}{lcl}
c_1&=&b_1b_2+\cdots+b_{Q-1}b_{Q}\\
&\vdots&\\
c_\tau &=&b_{(\tau -1)Q}b_{(\tau -1)Q+1}+\cdots+b_\ell b_{\ell+1}+\cdots+b_{\tau Q-1}b_{\tau Q}\\
&=&m_{\ell}+b_\ell b_{\ell+1}+\cdots+b_{\tau Q-1}b_{\tau Q}\\
&\vdots&\\
c_M&=&b_{(M-1)Q}b_{(M-1)Q+1}+\cdots+b_{N-1}b_N
\end{array}\right.
\end{eqnarray*}
Following the same argument as in Lemma~\ref{Lem1}, we obtain that the number of
all possible $\pm 1$-sequences of length $N$, whose first $\ell$ bits equal $b_1$,
$b_2$, $\ldots$, $b_\ell$ subject to \eqref{aa} is given by:
\begin{eqnarray*}
&&\binom{kQ-\ell}{(kQ-\ell+c_k-m_\ell)/2}\oneb\left\{|c_k-m_\ell|\leq kQ-\ell\right\}\\
&&\times\binom{Q}{(Q+c_{k+1})/2}\oneb\left\{|c_{k+1}|\leq Q\right\}\times
\cdots\times\binom{Q}{(Q+c_{M-1})/2}\oneb\left\{|c_{M-1}|\leq Q\right\}\\
&&\times\binom{N-(M-1)Q}{(N-(M-1)Q+c_M)/2}\oneb\left\{|c_{M}|\leq N-(M-1)Q\right\}.
\end{eqnarray*}
The proof is completed by noting that
$\ell=(\tau-1)Q+(\ell\bmod Q)$, $|c_k|\leq Q$ and $|c_M|\leq N-(M-1)Q$ are always valid.
\end{proof}

With the availability of the above lemma, the code construction algorithm in Section \ref{encoding}
can be performed.
Next, we re-derive the maximum-likelihood decoding metric for use of priority-first search
decoding algorithm. Continuing the derivation from \eqref{fastmetric} based on
$\Bb_k^T\Bb_k=\Gb_{\theta,k}$ for $1\leq k\leq M$ and $1\leq\theta\leq \Theta$, we can establish in terms of similar procedure as in
Section \ref{Rmmg} that:
\begin{eqnarray*}
\hat{\bb}
&=&\arg\min_{\bb\in \mathcal{C}}
\frac{1}{2}\sum_{k=1}^{M}\sum_{m=1}^{Q+P-1}\sum_{n=1}^{Q+P-1}
\left[-w_{m,n,k}^{(\theta)}b_{(k-1)Q-P+m+1}b_{(k-1)Q-P+n+1}\right]\oneb\{\bb\in\cc_\theta\}
\end{eqnarray*}
where for $1\leq m,n\leq Q+P-1$,
$$w_{m,n,k}^{(\theta)}=\sum_{i=0}^{P-1}\sum_{j=0}^{P-1}\delta_{i,j,k}^{(\theta)}
\mathrm{Re}\{\tilde y_{m+i,k} \tilde y_{n+j,k}^{\ast}\},$$
$\delta_{i,j,k}^{(\theta)}$ is the $(i,j)$th entry\footnote{
Under the assumption that $Q\geq P$,
the $i$th diagonal element of the target $\Gb_{\theta,1}$ is given by $Q-i+1$,
and the diagonal elements of the target $\Gb_{\theta,k}$ are
equal to $Q$ for $2\leq k<M$; hence, their inverse matrices exist.
However, when $P>N-(M-1)Q$, $\Gb_{\theta,M}$ has no inverse.
In such case, we re-define $\Db_{\theta,M}$ as:
$$\Db_{\theta,M}\triangleq\zerob_{[N-(M-1)Q]\times[N-(M-1)Q]}\oplus
\Gb_{\theta,M}^{-1}(N-(M-1)Q+1),$$
where $\Gb_{\theta,M}(j)$ is a $(P-j+1)\times(P-j+1)$ matrix
that contains the $j$th to $P$th rows and the $j$th to $P$th columns of $\Gb_{\theta,M}$.
} of $\Db_{\theta,k}=\Gb_{\theta,k}^{-1}$,
and
$\tilde\yb_k=[\zerob_{1\times (P-1)}\ \ \yb_k^H\ \ \zerob_{1\times(P-1)}]^H=[\tilde y_{1,k}\ \cdots
\ \tilde y_{Q+2P-2,k}]^T$.
As it turns out, the recursive on-the-fly metric for the priority-first search decoding algorithm is:
\begin{eqnarray*}
g(\bb_{(\ell)})-g(\bb_{(\ell-1)})
=\left\{\begin{array}{r}
\displaystyle\max_{1\leq\eta\leq\Theta}\alpha_{s,k}^{(\eta)}-b_{\ell}
\sum_{i=0}^{P-1}\sum_{j=0}^{P-1}
\delta_{i,j,k}^{(\theta)}\mathrm{Re}\{\tilde y_{s+i,k}\cdot
u_{j,k}(\bb_{(\ell)})\},\mbox{ for }P\leq s\leq Q\\
\displaystyle\max_{1\leq\eta\leq\Theta}\alpha_{r,k}^{(\eta)}
+\max_{1\leq\eta\leq\Theta}\alpha_{s,k+1}^{(\eta)}
-b_{\ell}\sum_{i=0}^{P-1}\sum_{j=0}^{P-1}
\bigg(\delta_{i,j,k}^{(\theta)}\mathrm{Re}\{\tilde y_{r+i,k}\cdot
u_{j,k}(\bb_{(\ell)})\}  \\
~+\delta_{i,j,k+1}^{(\theta)}\mathrm{Re}\{\tilde y_{s+i,k+1}\cdot
u_{j,k+1}(\bb_{(\ell)})\}\bigg),  \mbox{ otherwise.}
\end{array}
\right.
\end{eqnarray*}
where $-P+2\leq\ell\leq N$, $s=[(\ell+P-2)\bmod Q]+1$, $r=s+Q$,
$k=\max\{\lceil\ell/Q\rceil,1\}$,
$$
\alpha_{s,k}^{(\eta)}\triangleq
\sum_{n=1}^{s-1}\left|w_{s,n,k}^{(\eta)}\right|+\frac
12\left|w_{s,s,k}^{(\eta)}\right|$$ and
$$u_{j,k}(\bb_{(\ell+1)})\triangleq u_{j,k}(\bb_{(\ell)})
+\frac 12\left(b_\ell\tilde y_{s+j,k}^\ast+b_{\ell+1}\tilde
y_{s+j+1,k}^\ast\right)
$$
with initial values $g(\bb_{(-P+1)})=0$ and
$u_{j,k}(\bb_{((k-1)Q-P+2)})=0$ for $0\leq j\leq P-1$ and $1\leq
k\leq M$. In addition, the low-complexity heuristic function is
given by:
$$
h_2(\bb_{(\ell)})\triangleq \left\{\begin{array}{ll}
\displaystyle\sum_{m=s+1}^{Q+P-1}\max_{1\leq\eta\leq\Theta}\alpha_{m,k}^{(\eta)}
-\sum_{m=s+1}^{Q+P-1}\left|v_{m,k}^{(\theta)}(\bb_{(\ell)})\right|-\beta_{s,k}^{(\theta)}&\\
\hspace*{10mm}+\displaystyle\sum_{\kappa=k+1}^{M}\left(\sum_{m=1}^{Q+P-1}\max_{1\leq
\eta \leq \theta}
\alpha_{m,\kappa}^{(\eta)}-\beta_{0,\kappa}^{(\theta)}\right),&
\mbox{for }P\leq s\leq Q;\\
\displaystyle\sum_{m=s+1}^{Q+P-1}\max_{1\leq\eta\leq\Theta}\alpha_{m,k+1}^{(\eta)}
-\sum_{m=s+1}^{Q+P-1}\left|v_{m,k+1}^{(\theta)}(\bb_{(\ell)})\right|-\beta_{s,k+1}^{(\theta)}&\\
\hspace*{10mm}+\displaystyle\sum_{m=r+1}^{Q+P-1}\max_{1\leq\eta\leq\Theta}\alpha_{m,k}^{(\eta)}
-\sum_{m=r+1}^{Q+P-1}\left|v_{m,k}^{(\theta)}(\bb_{(\ell)})\right|-\beta_{r,k}^{(\theta)}&\\
\hspace*{10mm}+\displaystyle\sum_{\kappa=k+2}^{M}\left(\sum_{m=1}^{Q+P-1}\max_{1\leq
\eta \leq \theta}
\alpha_{m,\kappa}^{(\eta)}-\beta_{0,\kappa}^{(\theta)}\right),&
\mbox{otherwise},
\end{array}\right.
$$
where $s$, $r$ and $k$ are defined the same as for $g(\cdot)$,
$$v_{m,k}^{(\theta)}(\bb_{(\ell)})\triangleq
\sum_{n=1}^s
w_{m,n,k}^{(\theta)}b_{(k-1)Q+P+n-1}=v_{m,k}^{(\theta)}(\bb_{(\ell-1)})+w_{s,m,k}^{(\theta)}b_\ell,$$
and
$$\beta_{s,k}^{(\theta)}\triangleq
\sum_{m=s+1}^{Q+P-1}\left(\sum_{n=s+1}^{m-1}\left|w_{m,n,k}^{(\theta)}\right|
+\frac 12\left|w_{m,m,k}^{(\theta)}\right|\right)
=\beta_{s-1,k}^{(\theta)}-\sum_{n=s+1}^{Q+P-1}\left|w_{s,n,k}^{(\theta)}\right|
-\frac 12\left|w_{s,s,k}^{(\theta)}\right|$$ with initial values
$v_{m,k}^{(\theta)}(\bb_{(k-1)Q-P+2})=0$ and
$\beta_{0,k}^{(\theta)}=\sum_{m=1}^{Q+P-1}\alpha_{m,k}^{(\theta)}$.

It is worth mentioning that if the single-tree code is adopted,
$h_2(\cdot)$ can be further reduced to:
\begin{eqnarray*}
h_2(\bb_{(\ell)})&\triangleq&\left\{\begin{array}{ll}
\displaystyle\sum_{m=s+1}^{Q+P-1}\alpha_{m,k}^{(1)}
-\sum_{m=s+1}^{Q+P-1}\left|v_{m,k}^{(1)}(\bb_{(\ell)})\right|-\beta_{s,k}^{(1)}&\mbox{for }P\leq s\leq Q;\\
\displaystyle\sum_{m=s+1}^{Q+P-1}\alpha_{m,k+1}^{(1)}
-\sum_{m=s+1}^{Q+P-1}\left|v_{m,k+1}^{(1)}(\bb_{(\ell)})\right|-\beta_{s,k+1}^{(1)}&\\
\hspace*{10mm}+\displaystyle\sum_{m=r+1}^{Q+P-1}\alpha_{m,k}^{(1)}
-\sum_{m=r+1}^{Q+P-1}\left|v_{m,k}^{(1)}(\bb_{(\ell)})\right|-\beta_{r,k}^{(1)}&\mbox{otherwise},
\end{array}\right.
\end{eqnarray*}
since
$
\sum_{m=1}^{Q+P-1}\max_{1\leq \eta \leq \theta}
\alpha_{m,\kappa}^{(\eta)}-\beta_{0,\kappa}^{(\theta)}=\sum_{m=1}^{Q+P-1}\alpha_{m,\kappa}^{(1)}-\beta_{0,\kappa}^{(1)}=0;
$
hence, a sub-blockwise low-complexity on-the-fly
decoding can indeed be conducted under the single code tree
condition.

\begin{figure}[tbp]
\centering \includegraphics[width=5in]{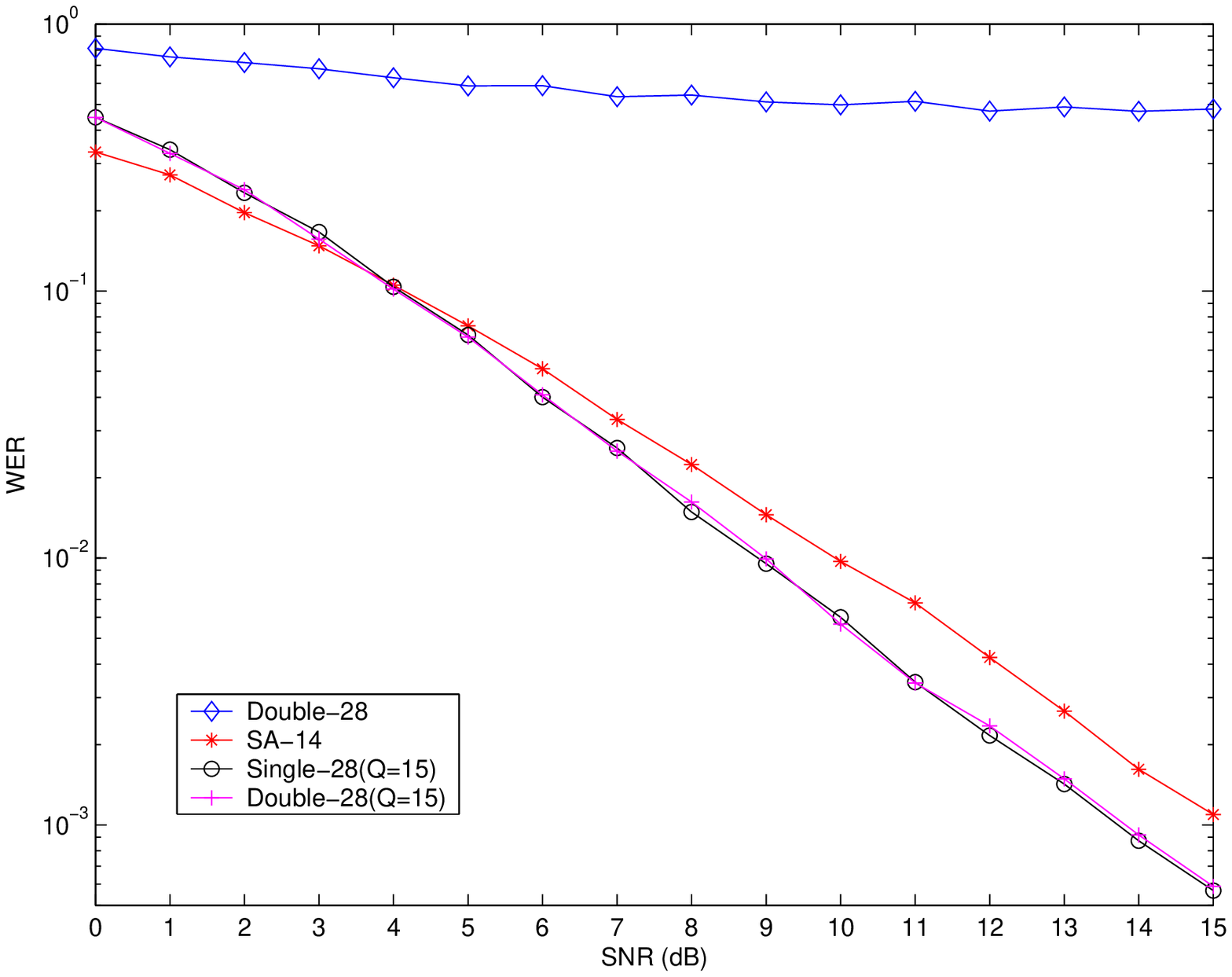} \caption{Word
error rates (BERs) for the codes of Double-$28$, SA-$14$,
Single-$28$($Q$=15) and Double-$28$($Q$=15)
over the quasi-static channel with $Q_{\rm chan}=15$.} \label{fig:sim10}
\end{figure}

\begin{figure}[tbp]
\centering \includegraphics[width=5in]{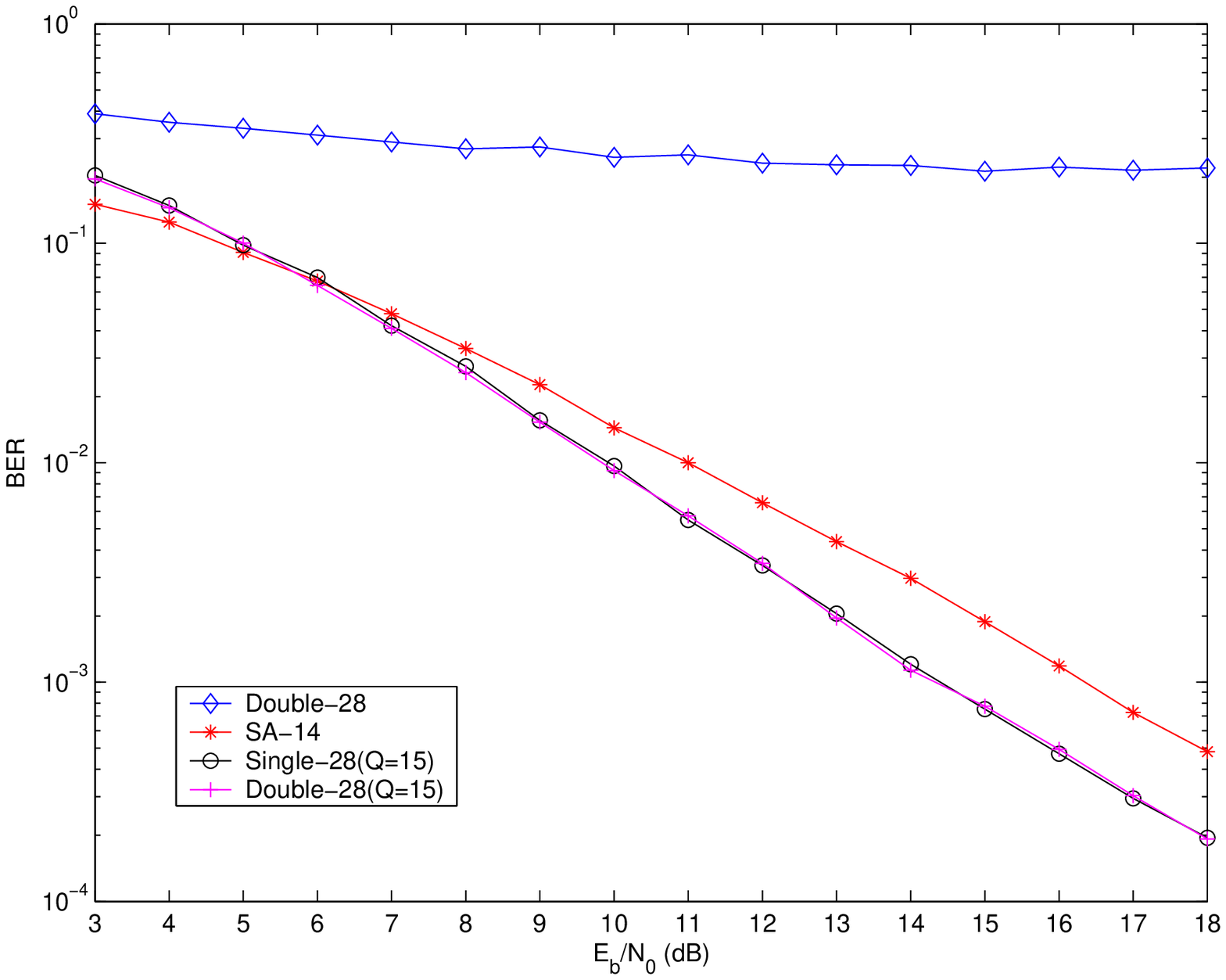} \caption{Bit
error rates (BERs) for the codes of Double-$28$, SA-$14$,
Single-$28$($Q$=15) and Double-$28$($Q$=15)
over the quasi-static channel with $Q_{\rm chan}=15$.} \label{fig:sim11}
\end{figure}

\begin{figure}[tbp]
\centering \includegraphics[width=5in]{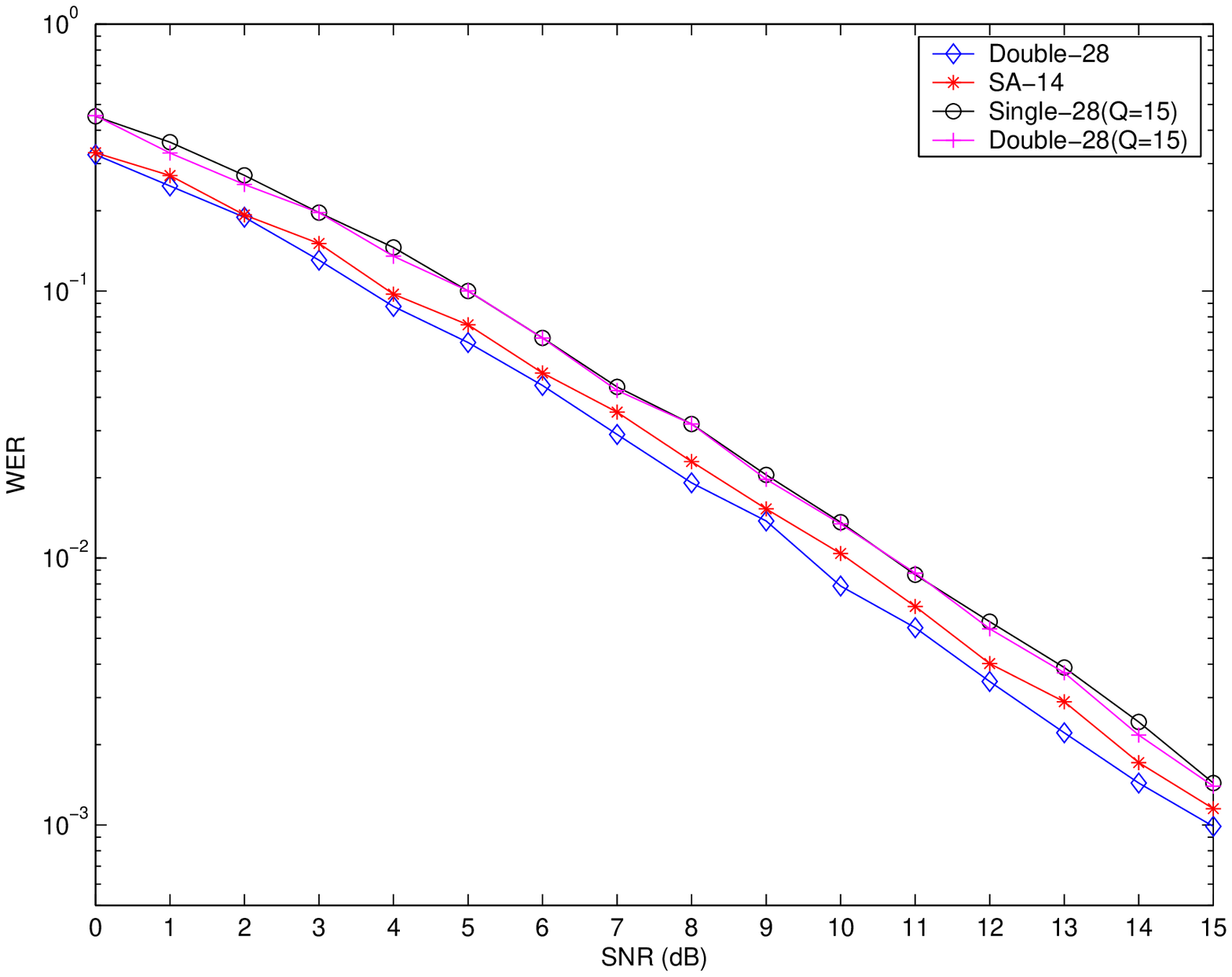} \caption{Word
error rates (BERs) for the codes of Double-$28$, SA-$14$,
Single-$28$($Q$=15) and Double-$28$($Q$=15)
over the quasi-static channel with $Q_{\rm chan}\geq 29$.} \label{fig:sim10a}
\end{figure}

\begin{figure}[tbp]
\centering \includegraphics[width=5in]{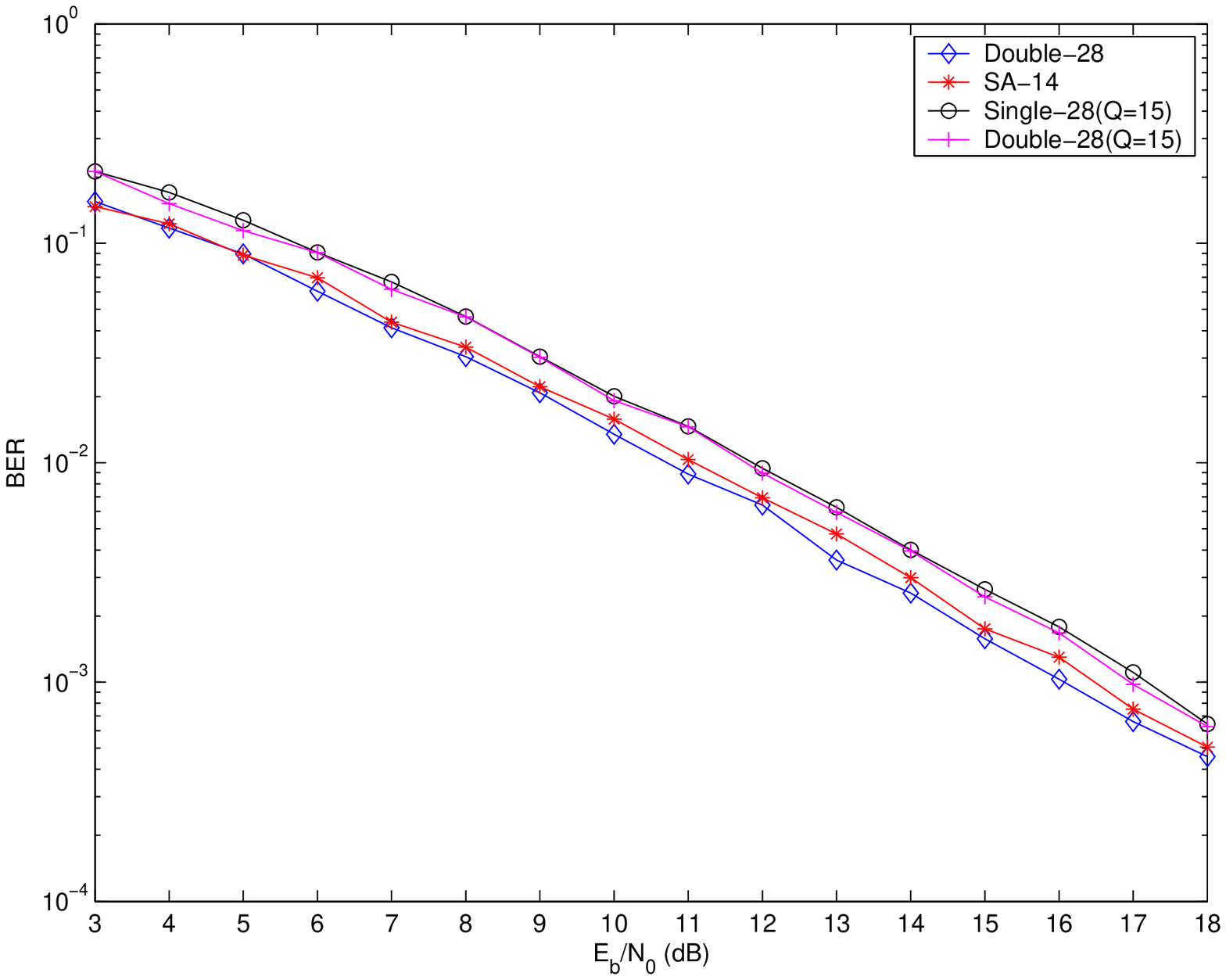} \caption{Bit
error rates (BERs) for the codes of Double-$28$, SA-$14$,
Single-$28$($Q$=15) and Double-$28$($Q$=15)
over the quasi-static channel with $Q_{\rm chan}\geq 29$.} \label{fig:sim11a}
\end{figure}

\begin{figure}[tbp]
\centering \includegraphics[width=5in]{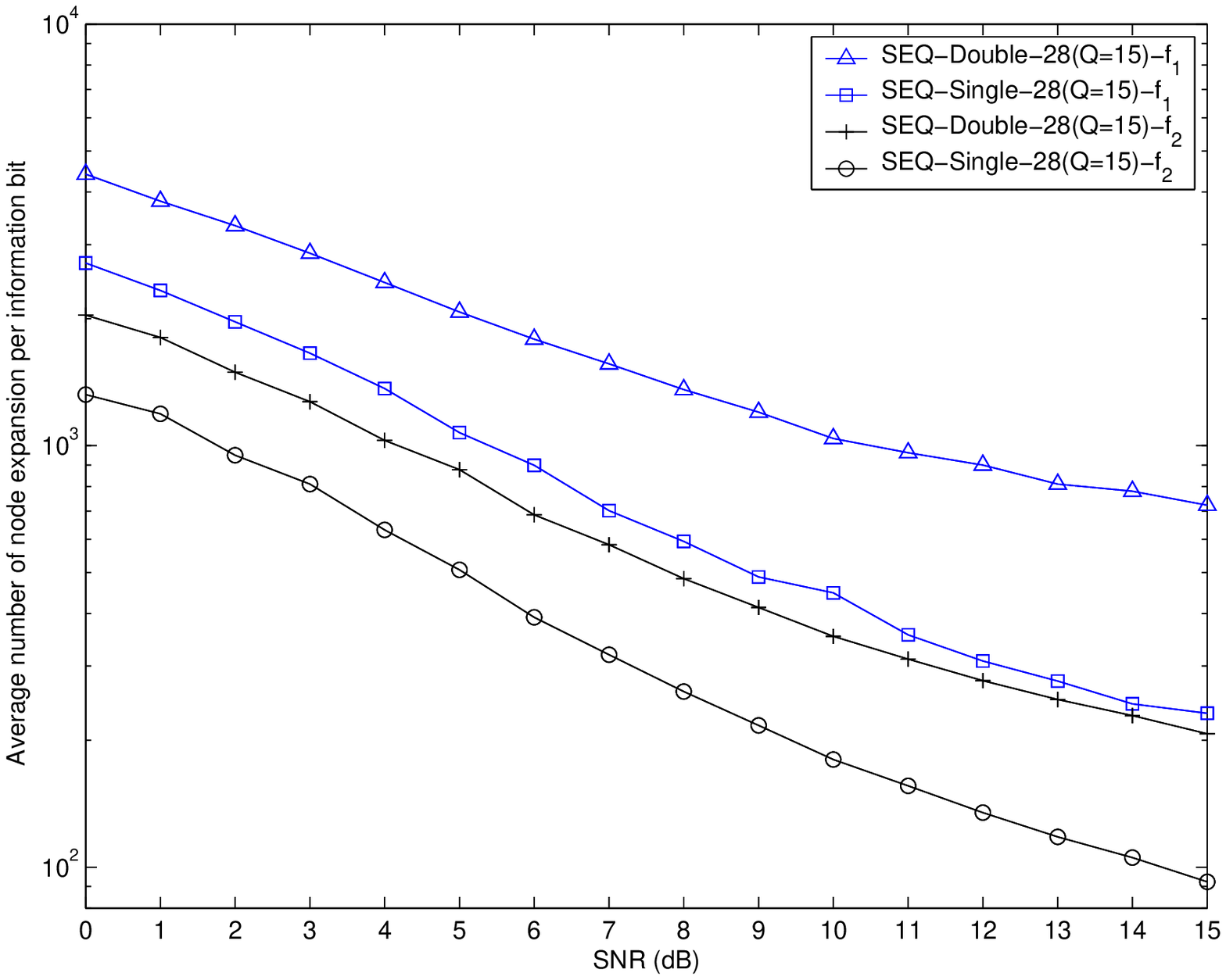}
\caption{
Average numbers of node expansions per information bit
for the codes of
Single-28($Q$=15) and Double-28($Q$=15)
using the priority-first search
decoding guided by either evaluation function $f_1$ or evaluation function
$f_2$ over the quasi-static channel with $Q_{\rm chan}=15$.} \label{fig:sim12}
\end{figure}

Figures \ref{fig:sim10} and \ref{fig:sim11} compare
four codes over fading channels whose channel coefficients
vary in every $15$-symbol period.
Notably, we will use $Q_{\rm chan}$ to denote the varying period of
the channel coefficients $\hb$, and retain $Q$ as the design parameter
for the nonlinear codes.
In notations, ``Double-$28$'' and ``SA-$14$'' denote the codes defined
in the previous sections, and ``Single-$28$($Q$=15)'' and
``Double-$28$($Q$=15)'' are the codes constructed based on the rule
introduced in this section under the design parameter $Q=15$.
Again, the mapping between the bit patterns
and codewords for the SA-$14$ code is defined by simulated annealing.

Both Figs.~\ref{fig:sim10} and \ref{fig:sim11} show that
the Double-$28$ code seriously degrades
when the channel coefficients unexpectedly vary in an intra-codeword fashion.
This hints that
the assumption that the channel coefficients remain
constant in a coding block is critical in the
code design in Section \ref{Codeconstruction}.
Figures \ref{fig:sim10a} and \ref{fig:sim11a} then indicate
that the codes taking into considerations the varying nature
of the channel coefficients within a codeword is robust in its performance when being
applied to channels with constant coefficients.
Thus, we may conclude that
for a channel whose coefficients vary more often than a coding block,
it is advantageous to use the code design for
a fast-fading environment considered in the section.

A more striking result from Fig.~\ref{fig:sim10}
is that even if the codeword length
of the Single-28($Q$=15) and the Double-28($Q$=15) codes is twice of
the SA-14 code, their word error rates are still markedly superior
at medium-to-high SNRs.
Note that the SA-14 code is the computer-optimized code specifically
for $Q_{\rm chan}=15$ channel.
This hints that when the channel memory order is known,
performance gain can be obtained
by considering the inter-subblock correlation, and favors a longer code design.

The decoding complexity, measured in terms of average number
of node expansions per information bit, for
codes of Single-28($Q$=15) and Double-28($Q$=15)
are illustrated in Fig.~\ref{fig:sim12}.
Similar observation is attained that the decoding metric $f_2$
yields less decoding complexity than
the on-the-fly decoding one $f_1$.

\section{Conclusions}
\label{Conclusions}

In this paper, we established the systematic rule to construct
codes based on the optimal signal-to-noise ratio framework
that requires every codeword to satisfy a ``self-orthogonal'' property to combat
the multipath channel effect. Enforced by
this structure, we can then derive a recursive maximum-likelihood
metric for the tree-based priority-first search decoding algorithm,
and hence, avoid the use of the time-consuming exhaustive decoder
that was previously used in \cite{Coskun,Giese,Skoglund} to decode
the structureless computer-optimized codes.
Simulations demonstrate
that the ruled-based codes we
constructed has almost identical performance to the computer-optimized
codes, but its decoding complexity, as anticipated,
is much lower than the exhaustive decoder.

Moreover, two maximum-likelihood decoding metrics were actually
proposed. The first one can be used in an on-the-fly fashion,
while the second one as having a much less decoding complexity
requires the knowledge of all channel outputs. The trade-off between them
is thus evident from our simulations.

Extensions of the code design to a fast-varying
quasi-static environment is added in Section \ref{fastfading}.
Although we only derive the coding rule and its decoding metric
for a fixed $Q$, further extension to the situation that
the channel coefficients $\hb$ vary non-stationarily
as the periods $Q_1$, $Q_2$, $\ldots$, $Q_M$ are not equal is
straightforward. Such design may be suitable
for, e.g.,
the frequency-hopping scheme
of Global System for Mobile communications (GSM) and Universal
Mobile Telecommunications System (UMTS), and the time-hopping
scheme in IS-$54$, in which cases
the channel coefficients change (or hop) at protocol-aware scheduled time instants
as similarly mentioned in \cite{Knopp}.

A limitation on the code design we proposed is that
the decoding complexity grows exponentially with the codeword length.
This constraint is owning to the tree structure of our constructed
codes. It will be an interesting and useful future work
to re-design the self-orthogonal codes that can be fit into a trellis structure,
and make them maximum-likelihood decodable by either the priority-first
search algorithm or the Viterbi-based algorithm.

\section*{Appendix}

\begin{lemma} Fix $P>1$ and $k\geq 1$.
For given integers $\qb=(q_1$, $q_2$, $\cdots$, $q_{P-1})\in[-k,k]^{P-1}$ and
$d_{2-P}$, $d_{3-P}$, $\cdots$, $d_0\in\{\pm 1\}$,
let the number of $d_1,d_2,\cdots,d_k$ that simultaneously satisfy
\begin{eqnarray*}
q_j&=&\sum_{i=1}^{k}d_{i-j}d_{i}\mbox{ for }1\leq j\leq P-1
\end{eqnarray*}
be denoted by $A_k(\qb|d_{2-P},\cdots,d_0)$.
Also, let $\Gb(\cb)$ be the $P\times P$ matrix of the Toeplitz form:
$$
\Gb(\cb)=
\begin{bmatrix}
N&c_1&c_2&\cdots&c_{P-1}\\
c_1&N&c_1&\cdots&c_{P-2}\\
c_2&c_1&N&\cdots&c_{P-3}\\
\vdots&\vdots&\vdots&\ddots&\vdots\\
c_{P-1}&c_{P-2}&c_{P-3}&\cdots&N
\end{bmatrix},
$$
where $\cb=(c_1,c_2,\cdots,c_{P-1})\in\{\pm 1\}^{P-1}$.
Then, for $\ell=1,2,\cdots,N$,
$$|\ac(\bb_{(\ell)}|\Gb(\cb))|=A_{N-\ell}\left(\cb-\mb_\ell\left|b_{\ell-P+2},\cdots,b_\ell\right.\right).$$
where $\mb_\ell=(m_\ell^{(1)},\cdots,m_\ell^{(P-1)})$ and
$m_\ell^{(k)}\triangleq(b_1b_{k+1}+\cdots+b_{\ell-k}b_\ell)\cdot\oneb\{\ell>k\}$.
\end{lemma}
\begin{proof}
For $\Bb^T\Bb=\Gb(\cb)$,
it requires
\begin{eqnarray*}
c_1&=&b_1b_2+b_2b_3+\cdots+b_{N-1}b_N=m_\ell^{(1)}+b_\ell b_{\ell+1}+\cdots+b_{N-1}b_N\\
c_2&=&b_1b_3+b_2b_4+\cdots+b_{N-2}b_N=m_\ell^{(2)}+b_{\ell-1}b_{\ell+1}+\cdots+b_{N-2}b_N\\
&\vdots&\\
c_{P-1}&=&b_1b_P+b_2b_{P+1}+\cdots+b_{N-P+1}b_N=m_\ell^{(P-1)}+b_{\ell-P+2}b_{\ell+1}+\cdots+b_{N-P+1}b_N.
\end{eqnarray*}
Re-writing the above equations as:
\begin{eqnarray*}
c_1-m_\ell^{(1)}&=&b_\ell b_{\ell+1}+\cdots+b_{N-1}b_N\\
c_2-m_\ell^{(2)}&=&b_{\ell-1}b_{\ell+1}+\cdots+b_{N-2}b_N\\
&\vdots&\\
c_{P-1}-m_\ell^{(P-1)}&=&b_{\ell-P+2}b_{\ell+1}+\cdots+b_{N-P+1}b_N,
\end{eqnarray*}
we obtain:
$$|\ac(\bb_{(\ell)}|\Gb(\cb))|=A_{N-\ell}\left(\cb-\mb_\ell\left|b_{\ell-P+2},\cdots,b_\ell\right.\right).$$
\end{proof}

\bigskip

It can be easily verified that $A_k(\qb|-d_{2-P},\cdots,-d_0)=A_k(\qb|d_{2-P},\cdots,d_0)$ since
$$
q_j=\sum_{i=1}^{j}(-d_{i-j})d_{i}+\sum_{i=j+1}^k d_{i-j}d_i=\sum_{i=1}^{j}d_{i-j}(-d_{i})+\sum_{i=j+1}^k (-d_{i-j})(-d_i).
$$
Therefore, only $2^{P-2}$ tables are required.
The tables of $A_k(\qb|d_{-P+2},\cdots,d_0)$ for $P=3$ and $1\leq k\leq 5$ are illustrated in Table~\ref{table3}
as an example.

\begin{table}[hbt]
\caption{Tables of $A_k(\qb|d_{-1},d_0)$ for $P=3$.
}\label{table3}
\begin{flushleft}
{\tiny
%%% k=1 %%%
\begin{tabular}{|c|c||c|c|c|}\hline
\multicolumn{2}{|c||}{}
&\multicolumn{3}{c|}{$q_2$}\\\cline{3-5}
\multicolumn{2}{|c||}{\raisebox{1.5ex}[0pt]{\!\!\!\!$A_1(\cdot|-1,-1)$\!\!\!\!}}&$-1$&$0$&$1$\\\hline\hline
     &$-1$&1&0&0\\\cline{2-5}
$q_1$&$ 0$&0&0&0\\\cline{2-5}
     &$ 1$&0&0&1\\\hline
\end{tabular}
\quad
\begin{tabular}{|c|c||c|c|c|}\hline
\multicolumn{2}{|c||}{}
&\multicolumn{3}{c|}{$q_2$}\\\cline{3-5}
\multicolumn{2}{|c||}{\raisebox{1.5ex}[0pt]{\!\!\!\!$A_1(\cdot|-1,1)$\!\!\!\!}}&$-1$&$0$&$1$\\\hline\hline
     &$-1$&0&0&1\\\cline{2-5}
$q_1$&$ 0$&0&0&0\\\cline{2-5}
     &$ 1$&1&0&0\\\hline
\end{tabular}\\[2mm]
%%% k=2 %%%
\begin{tabular}{|c|c||c|c|c|c|c|}\hline
\multicolumn{2}{|c||}{}
&\multicolumn{5}{c|}{$q_2$}\\\cline{3-7}
\multicolumn{2}{|c||}{\raisebox{1.5ex}[0pt]{\!\!\!\!$A_2(\cdot|-1,-1)$\!\!\!\!}}&$-2$&$-1$&$0$&$1$&$2$\\\hline\hline
     &$-2$&0&0&1&0&0\\\cline{2-7}
     &$-1$&0&0&0&0&0\\\cline{2-7}
$q_1$&$ 0$&1&0&1&0&0\\\cline{2-7}
     &$ 1$&0&0&0&0&0\\\cline{2-7}
     &$ 2$&0&0&0&0&1\\\hline
\end{tabular}
\quad
\begin{tabular}{|c|c||c|c|c|c|c|}\hline
\multicolumn{2}{|c||}{$$}
&\multicolumn{5}{c|}{$q_2$}\\\cline{3-7}
\multicolumn{2}{|c||}{\raisebox{1.5ex}[0pt]{\!\!\!\!$A_2(\cdot|-1,1)$\!\!\!\!}}&$-2$&$-1$&$0$&$1$&$2$\\\hline\hline
     &$-2$&0&0&0&0&1\\\cline{2-7}
     &$-1$&0&0&0&0&0\\\cline{2-7}
$q_1$&$ 0$&1&0&1&0&0\\\cline{2-7}
     &$ 1$&0&0&0&0&0\\\cline{2-7}
     &$ 2$&0&0&1&0&0\\\hline
\end{tabular}\\[2mm]
%%% k=3 %%%
\mbox{\begin{tabular}{|c|c||c|c|c|c|c|c|c|}\hline
\multicolumn{2}{|c||}{}
&\multicolumn{7}{c|}{$q_2$}\\\cline{3-9}
\multicolumn{2}{|c||}{\raisebox{1.5ex}[0pt]{\!\!\!\!$A_3(\cdot|-1,-1)$\!\!\!\!}}&$-3$&$-2$&$-1$&$0$&$1$&$2$&$3$\\\hline\hline
     &$-3$&  0 & 0 & 0 & 0 & 1 & 0 & 0  \\\cline{2-9}
     &$-2$&  0 & 0 & 0 & 0 & 0 & 0 & 0  \\\cline{2-9}
     &$-1$&  1 & 0 & 1 & 0 & 1 & 0 & 0  \\\cline{2-9}
$q_1$&$ 0$&  0 & 0 & 0 & 0 & 0 & 0 & 0  \\\cline{2-9}
     &$ 1$&  0 & 0 & 2 & 0 & 1 & 0 & 0  \\\cline{2-9}
     &$ 2$&  0 & 0 & 0 & 0 & 0 & 0 & 0  \\\cline{2-9}
     &$ 3$&  0 & 0 & 0 & 0 & 0 & 0 & 1  \\\hline
\end{tabular}
\quad
\begin{tabular}{|c|c||c|c|c|c|c|c|c|}\hline
\multicolumn{2}{|c||}{}
&\multicolumn{7}{c|}{$q_2$}\\\cline{3-9}
\multicolumn{2}{|c||}{\raisebox{1.5ex}[0pt]{\!\!\!\!$A_3(\cdot|-1,1)$\!\!\!\!}}&$-3$&$-2$&$-1$&$0$&$1$&$2$&$3$\\\hline\hline
     &$-3$&  0 & 0 & 0 & 0 & 0 & 0 & 1 \\\cline{2-9}
     &$-2$&  0 & 0& 0 &  0 & 0 & 0 & 0 \\\cline{2-9}
     &$-1$&  0 & 0 & 2 & 0 & 1 & 0 & 0 \\\cline{2-9}
$q_1$&$ 0$&  0 & 0 & 0 & 0 & 0 & 0 & 0 \\\cline{2-9}
     &$ 1$&  1 & 0 & 1 & 0 & 1 & 0 & 0 \\\cline{2-9}
     &$ 2$&  0 & 0 & 0 & 0 & 0 & 0 & 0 \\\cline{2-9}
     &$ 3$&  0 & 0 & 0 & 0 & 1 & 0 & 0 \\\hline
\end{tabular}}
$\vdots$\\[2mm]
%%% k=4 %%%
\mbox{\begin{tabular}{|c|c||c|c|c|c|c|c|c|c|c|}\hline
\multicolumn{2}{|c||}{}
&\multicolumn{9}{c|}{$q_2$}\\\cline{3-11}
\multicolumn{2}{|c||}{\raisebox{1.5ex}[0pt]{\!\!\!\!$A_4(\cdot|-1,-1)$\!\!\!\!}}&$-4$&$-3$&$-2$&$-1$&$0$&$1$&$2$&$3$&$4$\\\hline\hline
     &$-4$& 0 & 0 & 0 & 0 & 0 & 0 & 1 & 0 & 0 \\\cline{2-11}
     &$-3$& 0 & 0 & 0 & 0 & 0 & 0 & 0 & 0 & 0\\\cline{2-11}
     &$-2$& 0 & 0 & 2 & 0 & 1 & 0 & 1 & 0 & 0\\\cline{2-11}
     &$-1$& 0 & 0 & 0 & 0 & 0 & 0 & 0 & 0 & 0\\\cline{2-11}
$q_1$&$ 0$& 1 & 0 & 2 & 0 & 2 & 0 & 1 & 0 & 0\\\cline{2-11}
     &$ 1$& 0 & 0 & 0 & 0 & 0 & 0 & 0 & 0 & 0\\\cline{2-11}
     &$ 2$& 0 & 0 & 0 & 0 & 3 & 0 & 1 & 0 & 0\\\cline{2-11}
     &$ 3$& 0 & 0 & 0 & 0 & 0 & 0 & 0 & 0 & 0\\\cline{2-11}
     &$ 4$& 0 & 0 & 0 & 0 & 0 & 0 & 0 & 0 & 1\\\hline
\end{tabular}
\quad
\begin{tabular}{|c|c||c|c|c|c|c|c|c|c|c|}\hline
\multicolumn{2}{|c||}{}
&\multicolumn{9}{c|}{$q_2$}\\\cline{3-11}
\multicolumn{2}{|c||}{\raisebox{1.5ex}[0pt]{\!\!\!\!$A_4(\cdot|-1,1)$\!\!\!\!}}&$-4$&$-3$&$-2$&$-1$&$0$&$1$&$2$&$3$&$4$\\\hline\hline
     &$-4$& 0 & 0 & 0 & 0 & 0 & 0 & 0 & 0 & 1 \\\cline{2-11}
     &$-3$& 0 & 0 & 0 & 0 & 0 & 0 & 0 & 0 & 0 \\\cline{2-11}
     &$-2$& 0 & 0 & 0 & 0 & 3 & 0 & 1 & 0 & 0 \\\cline{2-11}
     &$-1$& 0 & 0 & 0 & 0 & 0 & 0 & 0 & 0 & 0 \\\cline{2-11}
$q_1$&$ 0$& 1 & 0 & 2 & 0 & 2 & 0 & 1 & 0 & 0 \\\cline{2-11}
     &$ 1$& 0 & 0 & 0 & 0 & 0 & 0 & 0 & 0 & 0 \\\cline{2-11}
     &$ 2$& 0 & 0 & 2 & 0 & 1 & 0 & 1 & 0 & 0 \\\cline{2-11}
     &$ 3$& 0 & 0 & 0 & 0 & 0 & 0 & 0 & 0 & 0 \\\cline{2-11}
     &$ 4$& 0 & 0 & 0 & 0 & 0 & 0 & 1 & 0 & 0 \\\hline
\end{tabular}}
$\vdots$\\[2mm]
%%% k=5 %%%
\mbox{\begin{tabular}{|c|c||c|c|c|c|c|c|c|c|c|c|c|}\hline
\multicolumn{2}{|c||}{}
&\multicolumn{11}{c|}{$q_2$}\\\cline{3-13}
\multicolumn{2}{|c||}{\raisebox{1.5ex}[0pt]{\!\!\!\!$A_5(\cdot|-1,-1)$\!\!\!\!}}&$-5$&$-4$&$-3$&$-2$&$-1$&$0$&$1$&$2$&$3$&$4$&$5$\\\hline\hline
     &$-5$& 0 & 0 & 0 & 0 & 0 & 0 & 0 & 0 & 1 & 0 & 0 \\\cline{2-13}
     &$-4$& 0 & 0 & 0 & 0 & 0 & 0 & 0 & 0 & 0 & 0 & 0 \\\cline{2-13}
     &$-3$& 0 & 0 & 0 & 0 & 3 & 0 & 1 & 0 & 1 & 0 & 0\\\cline{2-13}
     &$-2$& 0 & 0 & 0 & 0 & 0 & 0 & 0 & 0 & 0 & 0 & 0\\\cline{2-13}
     &$-1$& 1 & 0 & 2 & 0 & 4 & 0 & 2 & 0 & 1 & 0 & 0\\\cline{2-13}
$q_1$&$ 0$& 0 & 0 & 0 & 0 & 0 & 0 & 0 & 0 & 0 & 0 & 0\\\cline{2-13}
     &$ 1$& 0 & 0 & 3 & 0 & 3 & 0 & 3 & 0 & 1 & 0 & 0\\\cline{2-13}
     &$ 2$& 0 & 0 & 0 & 0 & 0 & 0 & 0 & 0 & 0 & 0 & 0\\\cline{2-13}
     &$ 3$& 0 & 0 & 0 & 0 & 0 & 0 & 4 & 0 & 1 & 0 & 0\\\cline{2-13}
     &$ 4$& 0 & 0 & 0 & 0 & 0 & 0 & 0 & 0 & 0 & 0 & 0\\\cline{2-13}
     &$ 5$& 0 & 0 & 0 & 0 & 0 & 0 & 0 & 0 & 0 & 0 & 1\\\hline
\end{tabular}
\quad
\begin{tabular}{|c|c||c|c|c|c|c|c|c|c|c|c|c|}\hline
\multicolumn{2}{|c||}{}
&\multicolumn{11}{c|}{$q_2$}\\\cline{3-13}
\multicolumn{2}{|c||}{\raisebox{1.5ex}[0pt]{\!\!\!\!$A_5(\cdot|-1,1)$\!\!\!\!}}&$-5$&$-4$&$-3$&$-2$&$-1$&$0$&$1$&$2$&$3$&$4$&$5$\\\hline\hline
     &$-5$& 0 & 0 & 0 & 0 & 0 & 0 & 0 & 0 & 0 & 0 & 1\\\cline{2-13}
     &$-4$& 0 & 0 & 0 & 0 & 0 & 0 & 0 & 0 & 0 & 0 & 0\\\cline{2-13}
     &$-3$& 0 & 0 & 0 & 0 & 0 & 0 & 4 & 0 & 1 & 0 & 0\\\cline{2-13}
     &$-2$& 0 & 0 & 0 & 0 & 0 & 0 & 0 & 0 & 0 & 0 & 0\\\cline{2-13}
     &$-1$& 0 & 0 & 3 & 0 & 3 & 0 & 3 & 0 & 1 & 0 & 0\\\cline{2-13}
$q_1$&$ 0$& 0 & 0 & 0 & 0 & 0 & 0 & 0 & 0 & 0 & 0 & 0\\\cline{2-13}
     &$ 1$& 1 & 0 & 2 & 0 & 4 & 0 & 2 & 0 & 1 & 0 & 0\\\cline{2-13}
     &$ 2$& 0 & 0 & 0 & 0 & 0 & 0 & 0 & 0 & 0 & 0 & 0\\\cline{2-13}
     &$ 3$& 0 & 0 & 0 & 0 & 3 & 0 & 1 & 0 & 1 & 0 & 0\\\cline{2-13}
     &$ 4$& 0 & 0 & 0 & 0 & 0 & 0 & 0 & 0 & 0 & 0 & 0\\\cline{2-13}
     &$ 5$& 0 & 0 & 0 & 0 & 0 & 0 & 0 & 0 & 1 & 0 & 0\\\hline
\end{tabular}}
}
\end{flushleft}
\end{table}

\section*{Acknowledgement}

The authors would like to thank Prof.~M.~Skoglund, Dr.~J.~Giese and Prof.~S.~Parkvall
of the Royal Institute of Technology (KTH), Stockholm,
Sweden, for kindly providing us their computer-searched codes
for further study in this paper.

\end{document}